\pdfoutput=1
\documentclass[11pt]{article}
\usepackage{subfig}
\usepackage{amssymb,amsmath,amsthm}
\usepackage{graphics,epsfig}
\usepackage{algorithmic}
\usepackage{algorithm}
\usepackage[margin=1 in]{geometry}

\usepackage{lipsum}
\usepackage{mathtools}
\usepackage{cuted}
\usepackage{hhline}

\newtheorem{defs}{Definition}
\newtheorem{lem}{Lemma}[section]

\newtheorem{thm}[lem]{Theorem}

\newtheorem{fact}[lem]{Fact}

\title{Approximation Algorithms for Scheduling Crowdsourcing Tasks in Mobile Social Networks}
\author{Chi-Yeh~Chen 
\\ Department of Computer Science and Information
Engineering, \\ National Cheng Kung University, \\
Taiwan, ROC. \\
chency@mail.csie.ncku.edu.tw.}

\begin{document}

\maketitle
\begin{abstract}
This paper addresses the scheduling problem in mobile social networks. We begin by proving that the approximation ratio analysis presented in the paper by Zhang \textit{et al.} (IEEE Transactions on Mobile Computing, 2025) is incorrect, and we provide the correct analysis results. Furthermore, when the required service time for a task exceeds the total contact time between the requester and the crowd worker, we demonstrate that the approximation ratio of the Largest-Ratio-First task scheduling algorithm can reach  $2 - \frac{1}{m}$. Next, we introduce a randomized approximation algorithm to minimize mobile social networks' total weighted completion time. This algorithm achieves an expected approximation ratio of $1.5 + \epsilon$ for $\epsilon>0$. Finally, we present a deterministic approximation algorithm that minimizes mobile social networks' total weighted completion time. This deterministic algorithm achieves an approximation ratio of  $\max\left\{2.5,1+\epsilon\right\}$ for $\epsilon>0$. Additionally, when the task's required service time or the total contact time between the requester and the crowd worker is sufficiently large, this algorithm can reach an approximation ratio of $1.5+\epsilon$ for $\epsilon>0$. 

\begin{keywords}
Mobile Crowdsourcing, Task Scheduling, Online Algorithm, Mobile Social Network
\end{keywords}
\end{abstract}

\section{Introduction}\label{sec:introduction}
With the advancement of technology, portable mobile devices have become widely used, allowing users to perform complex tasks such as measuring air quality~\cite{Dutta2009}, assessing urban noise~\cite{Rana2010}, monitoring road traffic~\cite{Farkas2014}, and evaluating surface conditions~\cite{Singh2017} through robust embedded sensors. The concept of \textit{Mobile crowdsourcing} emerges from the reality that a single mobile user's computational power and sensing capabilities are often insufficient for completing large-scale projects that consist of many smaller, independent tasks. Therefore, to efficiently achieve the overall project goal, seeking assistance from other users and distributing these tasks among them becomes necessary.

Conducting mobile crowdsourcing activities within a large-scale system often requires significant resources for management and maintenance. However, if a requester can distribute crowdsourcing tasks within their private social circle, they can leverage existing mobile social networks (MSNs) to carry out these tasks. This approach effectively reduces system overhead and saves costs. We denote the task requester as $u_{0}$ and the crowd worker as $u_{j}$, where $1 \leq j \leq m$. The requester $u_{0}$ assigns a specific task $s_{i}$ to a suitable crowd worker $u_{j}$. After completing the task, the crowd worker reports the results to the requester. The total time spent on completing the task includes the meeting times between $u_{0}$ and $u_{j}$, both before and after the task, in addition to the time $\tau_{i}$ required to complete the task itself. 

This paper focuses on task assignment and scheduling problems in mobile crowdsourcing systems based on Mobile Social Networks (MSN) to minimize the total weighted completion time. It is essential for there to be communication between the task requester and the crowd workers; this communication allows the requester to assign tasks and the crowd workers to provide feedback on those tasks. We consider a general model in which each crowdsourcing task requires one interaction with the crowd worker at the time of assignment and another interaction when the feedback is delivered. Additionally, we explore this problem in both offline and online settings.

\subsection{Related Work}
Research on mobile crowdsourcing focuses on two main areas: framework and application design~\cite{Guo2017, yuen2015, An2015}, and specific stages within the mobile crowdsourcing process. The latter includes aspects such as task scheduling~\cite{Xiao2015}, incentive mechanisms~\cite{Fan2015}, quality control~\cite{yan2015}, and issues related to security and privacy~\cite{Wu2014}.

There has been considerable literature on task scheduling problems~\cite{Allahverdi2008, Ding2016, Bridi2016, Bhatti2021}. Xiao \textit{et al.}~\cite{Xiao2015} explored MSN-based crowdsensing scheduling, which focuses on gathering environmental data using the crowd's mobile devices. To tackle this issue, they developed offline and online algorithms based on a greedy task assignment strategy to minimize the average completion time. Similarly, Zhang \textit{et al.}~\cite{Zhang2025} designed offline and online algorithms employing a greedy task assignment strategy to minimize both the total weighted completion time and the maximum completion time, also known as makespan. This paper adopts the system model proposed by Zhang \textit{et al.}~\cite{Zhang2025}, where the interaction between the task requester and the crowd workers is treated as a probabilistic event.

The parallel machine scheduling problem~\cite{Maiti2020} focuses on assigning tasks to multiple operating machines. This problem involves distributing tasks among various processing units and employing optimization algorithms to achieve objectives such as minimizing completion time or maximizing benefits. Unlike mobile crowdsourcing systems, which involve dynamic crowd workers whose locations and availability may change, the parallel machine scheduling problem deals with static machines or servers with fixed capabilities and locations. When we disregard communication time and focus solely on the capabilities of the machines, Graham~\cite{graham1969bounds} proposed a scheduling algorithm based on the Longest Processing Time (LPT) rule to minimize the maximum completion time. This approach guarantees an approximation ratio of $\frac{4}{3} - \frac{1}{3m}$. To minimize the total weighted completion time, Eastman \textit{et al.}~\cite{eastman1964} demonstrated that the Largest-Ratio-First (LRF) scheduling algorithm achieves an approximation ratio of 1.5. Furthermore, Kawaguchi and Seiki~\cite{Kawaguchi1986} improved this result by showing that the same method can achieve an approximation ratio of $\frac{(\sqrt{2}+1)}{2}\cong 1.207$.

In the context of the event-based scheduling network problem, She \textit{et al.}~\cite{She2015} introduced and defined the Utility-aware Social Event-participant Planning problem to provide personalized event arrangements for each participant. There is a wealth of similar research on the event-based scheduling network problem~\cite{Cheng2017, Cheng2021, Liu2012}. The goal of task scheduling within an event-based scheduling network typically focuses on creating an optimal activity itinerary for users to enhance their overall experience. This user-centric approach contrasts with the focus of this paper, which is centered on optimizing task allocation in a mobile crowdsourcing environment with an emphasis on minimizing time and costs.

\subsection{Our Contributions}
This paper addresses the scheduling problem in mobile social networks and presents various algorithms and their results. The specific contributions of this paper are outlined below:

\begin{itemize}
\item We demonstrate that the approximation ratio analysis in paper~\cite{Zhang2025} is incorrect and provide the correct analysis. Additionally, we prove that when the required service time for a task exceeds the total contact time between the requester and the crowd worker, the approximation ratio of the Largest-Ratio-First task scheduling algorithm can reach $2 - \frac{1}{m}$. These findings also update the results concerning the online algorithm presented in paper~\cite{Zhang2025}.

\item We propose a randomized approximation algorithm to minimize mobile social networks' total weighted completion time, achieving an expected approximation ratio of $1.5 + \epsilon$ for $\epsilon>0$.

\item We introduce a deterministic approximation algorithm to minimize mobile social networks' total weighted completion time. This algorithm achieves an approximation ratio of $\max\left\{2.5,1+\epsilon\right\}$ for $\epsilon>0$. Additionally, when the required service time for a task or the total contact time between the requester and the crowd worker is sufficiently large, the algorithm can achieve an improved approximation ratio of $1.5+\epsilon$ for $\epsilon>0$.
\end{itemize}

\subsection{Organization}
The structure of this paper is outlined as follows. In Section~\ref{sec:Preliminaries}, we introduce fundamental notations and preliminary concepts employed in subsequent sections. The main algorithms are presented in the following sections: Section~\ref{sec:LRF} introduce a Largest-Ratio-First based task scheduling algorithm, and Section~\ref{sec:Algorithm3} proposes a randomized algorithm with interval-indexed linear programming relaxation for minimizing the total weighted completion time, and Section~\ref{sec:Algorithm4} elaborates on the deterministic algorithm for the same. Section~\ref{sec:Results} compares the performance of the previous algorithm with that of the proposed algorithms. Finally, Section~\ref{sec:Conclusion} summarizes our findings and draws meaningful conclusions.

\section{Notation and Preliminaries}\label{sec:Preliminaries}
This section outlines the system model based on the approach proposed by Zhang \textit{et al.}~\cite{Zhang2025}. We consider a set of mobile users, denoted as $\mathcal{U} = \{u_{0}, u_{1}, \ldots, u_{m}\}$, within a mobile social network~\cite{Xiao2015}. Each user is equipped with a device that enables communication with others within a predefined range, offering sufficient connection duration and bandwidth to support crowdsourcing tasks. To model the communication patterns among mobile users, we adopt the mobility model presented in~\cite{Gao2009}, which assumes that the pairwise inter-contact times follow an exponential distribution. Specifically, the intermeeting time between any two distinct mobile users, $u_{i}$ and $u_{j}$ $(i \neq j)$, is modeled as an exponential random variable with a rate parameter $\lambda_{ij}$. This value, $\lambda_{ij}$, can be empirically estimated using historical communication data between users $u_{i}$ and $u_{j}$.

In the mobile social network, a user wants to recruit other users to help complete certain crowdsourcing tasks. We refer to the user who initiates the request as the "requester," while those who perform the tasks are known as "crowd workers." We denote the requester as $u_{0}$ and the other $m$ users in the mobile social network as $\left\{u_{1}, u_{2}, \ldots, u_{m}\right\}$, representing the crowd workers. Since communication between crowd workers is not permitted, we simplify our notation by denoting $\lambda_{0j}$ as $\lambda_{j}$.

The requester $u_{0}$ has a total of $n$ indivisible crowdsourcing tasks, denoted as $\mathcal{S}=\left\{s_{1}, s_{2}, \ldots, s_{n}\right\}$. The workload of each task $s_{i} \in \mathcal{S}$ is represented by its Required Service Time (RST), denoted as $\tau_{i}$. Each task $s_{i}$ is also associated with a weight $w_{i}$, which indicates the importance of the task. Since tasks are indivisible, each task can only be assigned to a single crowd worker, while each crowd worker is allowed to carry out multiple tasks. Additionally, each task requires extra time for distribution and feedback between the crowd worker and the requester. For convenience, we sometimes use $i$ to refer directly to $s_{i}$, and $j$ to refer to $u_{j}$.

A scheduling decision for $n$ crowdsourcing tasks involves assigning these tasks into $m$ disjoint sets, denoted as $\{\mathcal{S}_{1}, \mathcal{S}_{2}, \ldots, \mathcal{S}_{m}\}$, where each set $\mathcal{S}_{j}$ contains the tasks assigned to crowd worker $u_{j}$. Each crowd worker processes their assigned tasks sequentially until all tasks are completed. For a specific task $s_{i}$ assigned to crowd worker $u_{j}$, its Completion Time (CT) includes three components: (a) The time required for the initial meeting between the requester $u_{0}$ and crowd worker $u_{j}$ to complete the task distribution process. (b) The time taken by crowd worker $u_{j}$ to process all tasks in $\mathcal{S}_{j}$ that precede $s_{i}$, as well as the time taken for the task $s_{i}$ itself. (c) The time for the subsequent meeting between $u_{0}$ and $u_{j}$ after the feedback for $s_{i}$ has been received.

Since the exact meeting time between the requester and the crowd workers is unpredictable, we adopt the approach used by Zhang \textit{et al.}~\cite{Zhang2025} to define the Expected Meeting Time (EMT). Assuming that the inter-meeting time between $u_{0}$ and $u_{j}$ follows an exponential distribution, the EMT is calculated as $\frac{1}{\lambda_{j}}$. This value represents the time required in parts (a) and (c). Let $C_{i}$ denote the completion time of task $s_{i}$. This completion time consists of two expected meeting time intervals and the time required for $u_{j}$ to process all the tasks in $S_{j}$ that precede $s_{i}$, along with task $s_{i}$ itself. The objective is to schedule the tasks in a mobile social network in a way that minimizes the total weighted completion time of the tasks, which is represented by $\sum_{i=1}^{n} w_{i}C_{i}$.

The Expected Workload (EW) of a crowd worker is defined as follows.
\begin{defs}
\cite{Zhang2025}
(Expected Workload (EW)). The expected workload $EW_{j}$ of a crowd worker $u_{j}$ comprises three components:
(a) the expected meeting time for $u_{j}$ and $u_{0}$ to meet for the first time and complete the task distribution process,
(b) the total required service time to complete all tasks assigned to $u_{j}$, and
(c) the expected meeting time for $u_{j}$ and $u_{0}$ to meet again for task feedback delivery.
For instance, if a crowd worker $u_{j}$ is assigned the set of tasks $\mathcal{S}_{j}=\left\{s_{j_{1}}, s_{j_{2}}, \ldots, s_{j_{k}}\right\}$, the expected workload is given by $EW_{j}=\frac{2}{\lambda_{j}}+\tau_{j_{1}}+\tau_{j_{2}}+\cdots+\tau_{j_{k}}$. Moreover, if no tasks are assigned to $u_{j}$ (i.e., $\mathcal{S}_{j}$ is empty), we define $EW_{j}=\frac{2}{\lambda_{j}}$.
\end{defs}

It is worth noting that $C_{j_{k}}$ equals $EW_{j}$ when $s_{j_{k}}$ is the last task to be processed.
Let $\tau_{max}=\max_{i\in [1, n]} \tau_{i}$, $\tau_{min}=\min_{i\in [1, n]} \tau_{i}$, $\lambda_{max}=\max_{j\in [1, m]} \lambda_{j}$, $\lambda_{min}=\min_{j\in [1, m]} \lambda_{j}$, $w_{max}=\max_{i\in [1, n]} w_{i}$ and $w_{min}=\min_{i\in [1, n]} w_{i}$.
The notation and terminology used in this paper are summarized in Table~\ref{tab:notations}.

\begin{table}[!ht]
\caption{Notation and Terminology}
    \centering
        \begin{tabular}{||c|p{5in}||}
    \hline
		 $\mathcal{U}$      & Te set of mobile users $\mathcal{U}=\left\{u_{0}, u_{1}, \ldots, u_{m}\right\}$. \\
		\hline
     $m$      & The number of crowd workers.          \\
    \hline    
		 $\mathcal{S}$      & Te set of indivisible crowdsourcing tasks $\mathcal{S}=\left\{s_{1}, s_{2}, \ldots, s_{n}\right\}$. \\
		\hline
     $n$      & The number of tasks.          \\
    \hline    
  	 $\tau_{i}$ & Required Service Time (RST) of task $s_{i}$.         \\
    \hline
     $\lambda_{j}$      & Parameter of the exponential distribution for mobile users $u_{0}$ and $u_{j}$.        \\
    \hline
     $C_{i}$ & The completion time of task $s_{i}$.         \\
    \hline    
     $w_{i}$ & The weight of task $i$.         \\
    \hline
     $\tau_{max}, \tau_{min}$     & $\tau_{max}=\max_{i\in [1, n]} \tau_{i}$ and $\tau_{min}=\min_{i\in [1, n]} \tau_{i}$. \\
    \hline
     $\lambda_{max}, \lambda_{min}$     & $\lambda_{max}=\max_{j\in [1, m]} \lambda_{j}$ and $\lambda_{min}=\min_{j\in [1, m]} \lambda_{j}$. \\
    \hline  
     $w_{max}, w_{min}$     & $w_{max}=\max_{i\in [1, n]} w_{i}$ and $w_{min}=\min_{i\in [1, n]} w_{i}$. \\
    \hline		
		$L$, $\mathcal{L}$    & $L= \left\lceil \log_{(1+\eta)} \left(\sum_{i=1}^{n}\tau_{i}\right)\right\rceil$ and $\mathcal{L}=\left\{0, 1, \ldots, L\right\}$ be the set of time interval indices. \\
		\hline 			
			$I_{\ell}$ & 		$I_{\ell}$ is the $\ell$th time interval where $I_{0}=[0, 1]$ and $I_{\ell}=((1+\eta)^{\ell-1},(1+\eta)^{\ell}]$ for $1\leq \ell \leq L$. \\
		\hline 			
		$|I_{\ell}|$ & $|I_{\ell}|$ is the length of the $\ell$th interval where $|I_{0}|=1$ and $|I_{\ell}|=\eta(1+\eta)^{\ell-1}$ for $1\leq \ell \leq L$. \\
		\hline 						
		$e_{j}$ & $e_{j}=\frac{2}{\lambda_{j}}$ represents the total contact time between the requester $u_{0}$ and the crowd worker $u_{j}$. \\
		\hline 						
        \end{tabular}
    \label{tab:notations}
\end{table}

\section{Largest-Ratio-First Based Task Scheduling}\label{sec:LRF}
Zhang \textit{et al.}~\cite{Zhang2025} proposed a Largest-Ratio-First based task scheduling algorithm for crowdsourcing in mobile social networks. They demonstrated that the approximation ratio of their offline algorithm is $\frac{3}{2}$. In contrast, the competitive ratio of their online algorithm is $\frac{3}{2}+\frac{n\cdot w_{max}\sum_{j=1}^{m}\frac{3}{\lambda_{j}}}{w_{min}\sum_{i=1}^{n}\tau_{i}}$. The proof for this result is derived from modifications to the optimal scheduling bounds originally established by Eastman \textit{et al.}~\cite{eastman1964}. However, the modified bound is not guaranteed to hold, undermining their proof. The following theorem provides a counterexample showing that the bound does not necessarily hold.

\begin{thm}\label{thm:thm1}
Suppose that $WCT_{OPT}$ represents the minimum total weighted completion time produced by the optimal solution, the inequality 
\begin{eqnarray}\label{ineq:01}
WCT_{OPT}\geq \frac{1}{m}M_{1}+\frac{m-1}{2m}M_{n}+M_{\Lambda}
\end{eqnarray} 
does not necessarily hold where $M_{1}=\sum_{j=1}^{n}\sum_{i\leq j} w_{j}\tau_{i}$, $M_{n}=\sum_{i=1}^{n}w_{i}\tau_{i}$ and $M_{\Lambda}=\sum_{j=1}^{m}\sum_{s_{i}\in \mathcal{S}_{j}}w_{i}\cdot \frac{2}{\lambda_{j}}$.
\end{thm}
\begin{proof}
We consider four tasks $\left\{s_{1}, s_{2}, s_{3}, s_{4}\right\}$, along with their respective required service times and weights as shown in Table~\ref{tab:instance}. The requester distributes these four tasks among two crowd workers, where $\frac{2}{\lambda_{1}}=2$ and $\frac{2}{\lambda_{2}}=6$.
\begin{table}[!h]
\caption{An Instance of Task}
\centering
\begin{tabular}{c|ccc}
task    & $\tau_{i}$ & $w_{i}$ & $\frac{w_{i}}{\tau_{i}}$ \\ \hline
$s_{1}$ & 1          & 1       & 1                \\
$s_{2}$ & 1          & 1       & 1                \\
$s_{3}$ & 2          & 2       & 1                \\
$s_{4}$ & $T$        & $T$     & 1               
\end{tabular}
\label{tab:instance}
\end{table}
All four tasks have identical ratios, which means the allocation order under the Largest-Ratio-First rule can be arbitrary. As a result, the requester assigns the tasks to the two crowd workers in the order $\left\{s_{1}, s_{2}, s_{3}, s_{4}\right\}$ using the list algorithm, leading to $M_{\Lambda} = 6T + 8$. However, the optimal task order should be $\left\{s_{4},s_{1}, s_{2}, s_{3}\right\}$, assuming $T \geq 8$ . This results in the optimal workload completion time, $WCT_{OPT}=T^{2}+2T+35$. Moreover, we can obtain $M_{1}=T^{2}+4T+11$ and $M_{n}=T^{2}+6$. Substituting these into inequality~(\ref{ineq:01}) gives us: $T^{2}+2T+35\geq \frac{1}{2}(T^{2}+4T+11)+\frac{1}{4}(T^{2}+6)+6T+8$. After simplification, inequality $(T-4)(T-20)\geq 0$ is obtained. This inequality does not hold when  $4<T<20$, demonstrating that inequality~(\ref{ineq:01}) does not necessarily hold.
\end{proof}


\subsection{Offline Task Scheduling}
The algorithm introduced in Algorithm~\ref{Alg_LRF}, referred to as the Largest-Ratio-First (LRF) Algorithm, was proposed by Zhang \textit{et al.}~\cite{Zhang2025}. Without loss of generality, we assume $\frac{w_{1}}{\tau_{1}}\geq \frac{w_{2}}{\tau_{2}}\geq \cdots \geq \frac{w_{n}}{\tau_{n}}$ and $\lambda_{1}\geq \lambda_{2}\geq \cdots \geq \lambda_{m}$. In lines \ref{Alg_LRF:1}-\ref{Alg_LRF:2}, the initial values of all variables are set. Subsequently, in lines \ref{Alg_LRF:3}-\ref{Alg_LRF:4}, tasks are assigned to the least-loaded crowd workers following the list rule based on the Largest-Ratio-First approach.

\begin{algorithm}
\caption{The LRF Algorithm~\cite{Zhang2025}}
    \begin{algorithmic}[1]
		    \STATE $\mathcal{S}=\left\{s_{1},s_{2}, \ldots, s_{n}: \frac{w_{1}}{\tau_{1}}\geq \frac{w_{2}}{\tau_{2}}\geq \cdots \geq \frac{w_{n}}{\tau_{n}}\right\}$
				\STATE $\mathcal{U}=\left\{u_{1},u_{2}, \ldots, u_{m}:\lambda_{1},\lambda_{2}, \ldots, \lambda_{m}\right\}$
				\FOR{$j=1, 2, \ldots, m$} \label{Alg_LRF:1}
					\STATE $\mathcal{S}_{j}=\emptyset$
					\STATE $EW_{j}=\frac{2}{\lambda_{j}}$
				\ENDFOR                    \label{Alg_LRF:2}
				\FOR{$i=1, 2, \ldots, n$} \label{Alg_LRF:3}
				  \STATE $j_{min}=\arg\min\left\{EW_{k}|u_{k}\in U\right\}$
					\STATE $\mathcal{S}_{j_{min}}=\mathcal{S}_{j_{min}} \cup \left\{s_{i}\right\}$
					\STATE $EW_{j_{min}}=EW_{j_{min}}+\tau_{i}$
				\ENDFOR \label{Alg_LRF:4}
				\STATE \textbf{return} $\Lambda_{LRF}=\left\{\mathcal{S}_{1}, \mathcal{S}_{2}, \ldots, \mathcal{S}_{m}\right\}$
   \end{algorithmic}
\label{Alg_LRF}
\end{algorithm}

We use the following fact in later parts of the proof.
\begin{fact}\label{fact:1}
Given $a, a', b, b'>0$, we have $\frac{a+a'}{b+b'}\leq \max\left\{\frac{a}{b},\frac{a'}{b'}\right\}$.
\end{fact}

The following Theorem analyzes the approximation ratio of Algorithm~\ref{Alg_LRF}.
\begin{thm}\label{thm:thm2}
Let $WCT_{OPT}$ represent the minimum total weighted completion time from the optimal solution, and let $WCT_{LRF}$ denote the corresponding value obtained by Algorithm~\ref{Alg_LRF}. Therefore, it follows that
\begin{eqnarray*}\label{ineq:02}
\frac{WCT_{LRF}}{WCT_{OPT}} \leq \max\left\{\frac{3}{2},\frac{w_{max}\cdot\lambda_{max}}{w_{min}\cdot\lambda_{min}}\right\}.
\end{eqnarray*} 
\end{thm}
\begin{proof}
Based on the bounds derived in~\cite{eastman1964}, we have
\begin{eqnarray*}\label{ineq:03}
WCT_{LRF} < \frac{1}{m}M_{1}+\frac{m-1}{m}M_{n}+M_{\Lambda}
\end{eqnarray*} 
and 
\begin{eqnarray*}\label{ineq:04}
WCT_{OPT} \geq \frac{1}{m}M_{1}+\frac{m-1}{2m}M_{n}+M_{\Lambda^*}
\end{eqnarray*} 
where
\begin{eqnarray*}\label{ineq:05}
M_{1}=\sum_{j=1}^{n}\sum_{i\leq j} w_{j}\tau_{i}, M_{n}=\sum_{i=1}^{n}w_{i}\tau_{i},
\end{eqnarray*}
\begin{eqnarray*}\label{ineq:06}
M_{\Lambda}\leq n\cdot w_{max} \frac{2}{\lambda_{min}}, M_{\Lambda^*}\geq n\cdot w_{min} \frac{2}{\lambda_{max}}.
\end{eqnarray*}
Then we have 
\begin{eqnarray*}\label{ineq:07}
\frac{WCT_{LRF}}{WCT_{OPT}} & < &\frac{\frac{1}{m}M_{1}+\frac{m-1}{m}M_{n}+M_{\Lambda}}{\frac{1}{m}M_{1}+\frac{m-1}{2m}M_{n}+M_{\Lambda^*}} \\
 &\leq & \max\left\{\frac{\frac{1}{m}M_{1}+\frac{m-1}{m}M_{n}}{\frac{1}{m}M_{1}+\frac{m-1}{2m}M_{n}},\frac{M_{\Lambda}}{M_{\Lambda^*}}\right\}.
\end{eqnarray*}
The second inequality follows from Fact~\ref{fact:1}. According to~\cite{eastman1964}, we have
\begin{eqnarray}\label{ineq:08}
\frac{\frac{1}{m}M_{1}+\frac{m-1}{m}M_{n}}{\frac{1}{m}M_{1}+\frac{m-1}{2m}M_{n}}\leq \frac{3}{2}.
\end{eqnarray}
We also have 
\begin{eqnarray}\label{ineq:09}
\frac{M_{\Lambda}}{M_{\Lambda^*}} \leq \frac{w_{max}\cdot\lambda_{max}}{w_{min}\cdot\lambda_{min}}.
\end{eqnarray}
By combining inequalities~(\ref{ineq:08}) and (\ref{ineq:09}), we obtain 
\begin{eqnarray*}\label{ineq:10}
\frac{WCT_{LRF}}{WCT_{OPT}} \leq \max\left\{\frac{3}{2},\frac{w_{max}\cdot\lambda_{max}}{w_{min}\cdot\lambda_{min}}\right\}.
\end{eqnarray*} 
\end{proof}

Next, we will demonstrate that Algorithm~\ref{Alg_LRF} can achieve a better approximation ratio, provided that the required service time for a task exceeds the total contact time between the requester and the crowd worker. According to Smith~\cite{smith1956}, a schedule achieves optimality if and only if the jobs are arranged in non-increasing order based on the ratio $\frac{w_i}{\tau_i}$. In this context, we examine the problem of assigning $n$ tasks to a crowd worker, where the actual inter-meeting time and feedback delivery time are represented by $\sum_{j=1}^{m}\frac{2}{\lambda_{j}}$. Let $D_{i}^{*}$ denote the optimal completion time of task $i$ in this scenario. We have
\begin{eqnarray}\label{ineq:17}
D_{i}^{*} = \sum_{j=1}^{m}\frac{2}{\lambda_{j}} + \sum_{i'\leq i} \tau_{i'}.
\end{eqnarray}

We are examining the problem of assigning $n$ tasks to $m$ crowd workers. Each crowd worker has an actual inter-meeting time and feedback delivery time represented by $\frac{2}{\lambda_{j}}$ for all $j \in [1, m]$. We also know that $\tau_{min} > \frac{2}{\lambda_{max}}$. Since this condition holds, every crowd worker will be assigned at least one task. Let $C_{i}^{*}$ denote the optimal task completion time for task $i$ in this context. For all tasks $i \in [1, n]$, we can establish that $\frac{D_{i}^{*}}{m} \leq C_{i}^{*}$. This inequality leads us to the following theorem.

\begin{thm}\label{thm:thm55}
Given the assumption that $\tau_{min} > \frac{2}{\lambda_{max}}$ and $n \geq m$, the Algorithm~\ref{Alg_LRF} achieves an approximation ratio of at most $2 - \frac{1}{m}$.
\end{thm}
\begin{proof}
Since $\tau_{min} > \frac{2}{\lambda_{max}}$, every crowd worker will be assigned at least one task. Let $C_{i}$ represent the completion time of task $i$ for the LRF algorithm. We have 
\begin{eqnarray}\label{ineq:18}
C_{i} &\leq& \frac{\sum_{j=1}^{m}\frac{2}{\lambda_{j}} + \sum_{i'< i} \tau_{i'}}{m}+\tau_{i} \notag\\
      & =  & \frac{\sum_{j=1}^{m}\frac{2}{\lambda_{j}} + \sum_{i'\leq i} \tau_{i'}}{m}+\left(1-\frac{1}{m}\right)\tau_{i} \notag\\
			& \leq  & C_{i}^*+\left(1-\frac{1}{m}\right)\tau_{i} \label{ineq:18_1}\\
			& \leq  & C_{i}^*+\left(1-\frac{1}{m}\right)C_{i}^* \label{ineq:18_2}\\
			& =   & \left(2-\frac{1}{m}\right)C_{i}^*. \notag
\end{eqnarray}
for all $i\in [1,n]$ where inequality (\ref{ineq:18_1}) derives from inequality (\ref{ineq:17}), and inequality (\ref{ineq:18_2}) is based on the fact that $\tau_i \leq C_i^*$.
Thus, Theorem~\ref{thm:thm55} holds.
\end{proof}

\subsection{Online Task Scheduling}
This section focuses on the scheduling problem in an online setting. In this scenario, the requester gradually meets crowd workers over time. When the requester, denoted as $u_{0}$, meets a crowd worker $u_{j}$, scheduling is performed based on the currently unassigned crowd workers and tasks. The tasks assigned to $u_{j}$ are then delivered to it for processing. Since the requester has already met this crowd worker, estimating the intermeeting time for scheduling is unnecessary. Therefore, the expected workload for crowd worker $u_{j}$ is given by $EW_{j} = \frac{1}{\lambda_{j}}$. Meanwhile, the expected workloads for the other crowd workers are updated to $EW_{j'}$ are updated to $EW_{j'} = \frac{2}{\lambda_{j'}} - \frac{1}{\lambda_{j}}$. The algorithm introduced in Algorithm~\ref{Alg_ONLRF}, known as the Cost Minimized Online Scheduling (CosMOS) Algorithm, was proposed by Zhang \textit{et al.}~\cite{Zhang2025}.

We modify Line~\ref{Alg_ONLRF:1} by changing the equation $EW_{k} = \frac{2}{\lambda_{k}}$ to $EW_{k} = \frac{2}{\lambda_{k}} - \frac{1}{\lambda_{j}}$. Additionally, we introduce Line~\ref{Alg_ONLRF:2} to specify that the task is assigned to $u_{j}$ for processing at this stage, making the algorithm more clearly defined. In Line~\ref{Alg_ONLRF:1}, the approach set by Zhang \textit{et al.}~\cite{Zhang2025} tends to favor task assignments towards $u_{j}$. Our modification, however, allows for a more equitable distribution of tasks.

Since $\lambda_{1} \geq \lambda_{2} \geq \cdots \geq \lambda_{m}$, the requester $u_{0}$ assigns tasks in the order of $\left\{u_{1}, u_{2}, \ldots, u_{m}\right\}$. Let $WCT_{0}$ represent the total weighted completion time of the offline algorithm, which we denote as $WCT_{LRF} = WCT_{0}$. After $u_{0}$ interacts with $u_{j}$ in Algorithm~\ref{Alg_ONLRF}, we define $WCT_{j}$ as the total weighted completion time at that point. Consequently, the total weighted completion time after all assigned tasks is $WCT_{m}$. This definition allows us to conclude that $WCT_{CosMOS} = WCT_{m}$.

\begin{algorithm}
\caption{The CosMOS Algorithm~\cite{Zhang2025}}
    \begin{algorithmic}[1]
		    \STATE $\mathcal{S}=\left\{s_{1},s_{2}, \ldots, s_{n}: \frac{w_{1}}{\tau_{1}}\geq \frac{w_{2}}{\tau_{2}}\geq \cdots \geq \frac{w_{n}}{\tau_{n}}\right\}$
				\STATE $\mathcal{U}=\left\{u_{1},u_{2}, \ldots, u_{m}:\lambda_{1},\lambda_{2}, \ldots, \lambda_{m}\right\}$
				\WHILE{$u_{0}$ meets $u_{j}$}
					\STATE $\mathcal{S}_{j}=\emptyset$
					\STATE $EW_{j}=\frac{1}{\lambda_{j}}$
					\FOR{each $k$ such that $u_{k}\in U\setminus\left\{u_{j}\right\}$}
						\STATE $\mathcal{S}_{k}=\emptyset$
						\STATE $EW_{k}=\frac{2}{\lambda_{k}}-\frac{1}{\lambda_{j}}$ \label{Alg_ONLRF:1}
					\ENDFOR
					\FOR{each $i$ such that $s_{i}\in S$ (in an increasing order of $i$)}
						\STATE $k_{min}=\arg\min\left\{EW_{k}|u_{k}\in U\right\}$
						\STATE $\mathcal{S}_{k_{min}}=\mathcal{S}_{k_{min}} \cup \left\{s_{i}\right\}$
						\STATE $EW_{k_{min}}=EW_{k_{min}}+\tau_{i}$
					\ENDFOR
					\STATE Assign $\mathcal{S}_{j}$ to $u_{j}$ and schedule its execution following the LRF order. \label{Alg_ONLRF:2}
					\STATE $\mathcal{S}=\mathcal{S}\setminus \mathcal{S}_{j}$
					\STATE $U=U\setminus \left\{u_{j}\right\}$
				\ENDWHILE
				\STATE \textbf{return} $\Lambda_{CosMOS}=\left\{\mathcal{S}_{1}, \mathcal{S}_{2}, \ldots, \mathcal{S}_{m}\right\}$
   \end{algorithmic}
\label{Alg_ONLRF}
\end{algorithm}

The following Theorem, proved by Zhang \textit{et al.}~\cite{Zhang2025}, shows that the total weighted completion time is non-increasing over successive decision steps.
\begin{thm}\label{thm:thm4}
\cite{Zhang2025} In the CosMOS algorithm, the total weighted completion time (WCT) is non-increasing over successive decision steps. Specifically, after each step, the WCT either remains unchanged or decreases; that is, $WCT_{LRF}=WCT_{0}\geq WCT_{1}\geq \cdots \geq WCT_{m}=WCT_{CosMOS}$.
\end{thm}

The following theorem analyzes the competitive ratio of Algorithm~\ref{Alg_ONLRF}.
\begin{thm}\label{thm:thm3}
Suppose an entity with full knowledge of the mobility of all crowd workers exists, specifically knowing the exact times of each meeting between the requester and the crowd workers in advance. Using this information, the entity can make optimal online task scheduling decisions, represented as $\Lambda_{OPT}=\left\{\mathcal{S}_{1}^{*}, \mathcal{S}_{2}^{*}, \ldots, \mathcal{S}_{m}^{*}\right\}$. Then we have 
\begin{eqnarray*}\label{ineq:11}
\frac{WCT_{CosMOS}}{WCT_{OPT}} \leq \alpha\left(1 + \frac{n\cdot w_{max}\cdot\frac{2}{\lambda_{min}}}{w_{min}\sum_{i=1}^{n}\tau_{i}}\right)
\end{eqnarray*} 
where $\alpha =\max\left\{\frac{3}{2},\frac{w_{max}\cdot\lambda_{max}}{w_{min}\cdot\lambda_{min}}\right\}$ and $\alpha = 2-\frac{1}{m}$ if $\tau_{min} > \frac{2}{\lambda_{max}}$.
\end{thm}
\begin{proof}
The following proof follows the steps in~\cite{Zhang2025}, with the key differences being the use of our newly proven Theorems~\ref{thm:thm2}, \ref{thm:thm55}, and a modification of their overly relaxed bound. Given an arbitrary scheduling decision $\Lambda=\left\{\mathcal{S}_{1}, \mathcal{S}_{2}, \ldots, \mathcal{S}_{m}\right\}$, we have
\begin{eqnarray*}\label{ineq:12}
WCT_{\Lambda}=\sum_{j=1}^{m}\sum_{s_{i}\in \mathcal{S}_{j}} w_{i}(t_{j}+t'_{j}+T_{i}+\tau_{i}),
\end{eqnarray*} 
where $t_{j}$ and $t'_{j}$ denote the actual inter-meeting times between $u_{0}$ and $u_{j}$ for task distribution and feedback delivery, respectively, and $T_{i}$ represents the total required service time of the tasks scheduled before $s_{i}$. 

We also have
\begin{eqnarray*}\label{ineq:13}
WCT_{OPT}=\sum_{j=1}^{m}\sum_{s_{i}\in \mathcal{S}_{j}^{*}} w_{i}(t_{j}+t'_{j}+T_{i}+\tau_{i}).
\end{eqnarray*} 

Let $\Lambda'$ denote the scheduling decision of the offline version of CosMOS obtained by adopting the scheduling policy of $\Lambda_{OPT}$. Then, we have 
\begin{eqnarray*}\label{ineq:14}
WCT'& = & \sum_{j=1}^{m}\sum_{s_{i}\in \mathcal{S}_{j}^{*}} w_{i}\left(\frac{2}{\lambda_{j}}+T_{i}+\tau_{i}\right) \\
    & = & WCT_{OPT} +\sum_{j=1}^{m}\sum_{s_{i}\in \mathcal{S}_{j}^{*}} w_{i}\left(\frac{2}{\lambda_{j}}-t_{j}-t'_{j}\right) \\
		& \leq  & WCT_{OPT} +\sum_{j=1}^{m}\sum_{s_{i}\in \mathcal{S}_{j}^{*}} w_{i}\cdot\frac{2}{\lambda_{j}} \\
		& \leq  & WCT_{OPT} +n\cdot w_{max}\cdot\frac{2}{\lambda_{min}}.
\end{eqnarray*} 

According to Theorems~\ref{thm:thm2}, \ref{thm:thm55} and~\ref{thm:thm4}, we have
\begin{eqnarray*}\label{ineq:15}
WCT_{CosMOS} & \leq & WCT_{LRF} \\
             & \leq & \alpha \cdot WCT' \\
						 & \leq & \alpha \cdot WCT_{OPT}+ \alpha \cdot n\cdot w_{max}\cdot\frac{2}{\lambda_{min}}
\end{eqnarray*}
where $\alpha =\max\left\{\frac{3}{2},\frac{w_{max}\cdot\lambda_{max}}{w_{min}\cdot\lambda_{min}}\right\}$ and $\alpha = 2-\frac{1}{m}$ if $\tau_{min} > \frac{2}{\lambda_{max}}$.

Since $WCT_{OPT}\geq w_{min}\sum_{i=1}^{n}\tau_{i}$, we have
\begin{eqnarray*}\label{ineq:16}
\frac{WCT_{CosMOS}}{WCT_{OPT}} & \leq & \alpha\left(1 + \frac{n\cdot w_{max}\cdot\frac{2}{\lambda_{min}}}{w_{min}\sum_{i=1}^{n}\tau_{i}}\right).
\end{eqnarray*}
Thus, Theorem~\ref{thm:thm3} holds.
\end{proof}

\section{Randomized Approximation Algorithm with Interval-indexed Approach}\label{sec:Algorithm3}
This section proposes a randomized approximation algorithm for the crowdsourcing scheduling problem. 
The proposed approach is inspired by Schulz and Skutella~\cite{Schulz2002}.
This approach can be modified into a deterministic approximation algorithm, as discussed in section~\ref{sec:Algorithm4}. For a given positive parameter $\eta$, let 
\begin{eqnarray*}
L= \left\lceil \log_{(1+\eta)} \left(\sum_{i=1}^{n}\tau_{i}\right)\right\rceil.
\end{eqnarray*}
Consequently, the value of $L$ is bounded polynomially by the input size of the scheduling problem being considered. Let $\mathcal{L}=\left\{0, 1, \ldots, L\right\}$ represent the set of time interval indices.
Let $I_{\ell}$ be the $\ell$th time interval where $I_{0}=[0, 1]$ and $I_{\ell}=((1+\eta)^{\ell-1},(1+\eta)^{\ell}]$ for $1\leq \ell \leq L$. The length of the $\ell$th interval, denoted as $|I_{\ell}|$, is defined as follows: $|I_{0}|=1$ and for $1\leq \ell \leq L$, $|I_{\ell}|=\eta(1+\eta)^{\ell-1}$. Let $y_{ij\ell}$ be a variable where $i\in [1,n]$, $j\in [1,m]$ and $\ell\in \mathcal{L}$. The expression $y_{ij\ell}\cdot |I_{\ell}|$ represents the amount of time that task $s_{i}$ is processed during the time interval $I_{\ell}$ by crowd worker $u_{j}$. Furthermore, the expression $(y_{ij\ell}\cdot |I_{\ell}|)/\tau_{i}$ represents the fraction of task $s_{i}$ that is completed within interval $I_{\ell}$ by crowd worker $u_{j}$. Let $e_{j}=\frac{2}{\lambda_{j}}$ denote the total contact time between the requester $u_{0}$ and crowd worker $u_{j}$. Now, we can consider the following linear program, which is based on these interval-indexed variables:

\begin{subequations}\label{coflow:interval}
\begin{align}
& \text{min}  && \sum_{i \in [1, n]} w_{i} C_{i}     &   & \tag{\ref{coflow:interval}} \\
& \text{s.t.} && \sum_{j=1}^{m}\sum_{\ell=0}^{L} \frac{y_{ij\ell}\cdot |I_{\ell}|}{\tau_{i}} = 1, && \forall i\in [1, n] \label{interval:a} \\
&  && \sum_{i \in [1, n]} y_{ij\ell} \leq 1, && \forall j\in [1, m], \forall \ell\in \mathcal{L} \label{interval:b} \\
&  && C_{i} = D_{i}, && \forall i\in [1, n] \label{interval:d} \\
&  && C_{i} \geq \sum_{j=1}^{m}\sum_{\ell=0}^{L}y_{ij\ell}\cdot|I_{\ell}|, && \forall i\in [1, n] \label{interval:e} \\
&  && y_{ij\ell} \geq 0,&& \forall i\in [1, n], \forall j\in [1, m],\forall \ell\in \mathcal{L}\label{interval:g} 
\end{align}
\end{subequations}
where
\begin{eqnarray}\label{interval:d2}
D_{i} = \sum_{j=1}^{m}\sum_{\ell=0}^{L}\left(\frac{y_{ij\ell}\cdot |I_{\ell}|}{\tau_{i}}\left(e_{j}+(1+\eta)^{\ell-1}\right)+\frac{1}{2}y_{ij\ell}\cdot |I_{\ell}|\right).
\end{eqnarray}
To clarify the notation in equation (\ref{interval:d2}), we define $(1+\eta)^{\ell-1}$ to equal $1/2$ when $\ell=0$ in the expression $\frac{y_{ij\ell} \cdot |I_{\ell}|}{\tau_{i}} \left(e_{j} + (1+\eta)^{\ell-1}\right)$. This definition promotes consistent notation for all values of $\ell$.
The constraint (\ref{interval:a}) ensures that the processing requirements for each task are properly fulfilled. The crowd worker capacity constraint (\ref{interval:b}) stipulates that each crowd worker can handle only one task at a time. Additionally, constraints (\ref{interval:d}) and (\ref{interval:e}) establish lower bounds on the completion time for task $s_{i}$.
Furthermore, the equation (\ref{interval:d2}) represents the lower bound of the completion time for task $s_{i}$. This lower bound occurs during continuous transmission between $C_{i} - \tau_{i}$ and $C_{i}$. It also accounts for each crowd worker's fraction of the total contact time for task $s_{i}$, that is, $\sum_{j=1}^{m}\sum_{\ell=0}^{L}\left(\frac{y_{ij\ell}\cdot |I_{\ell}|}{\tau_{i}}e_{j}\right)$.

\begin{algorithm}
\caption{The RIS Algorithm}
    \begin{algorithmic}[1]
				\STATE Compute an optimal solution $y$ to linear programming (\ref{coflow:interval}). \label{alg3-2}
				\STATE For all tasks $s_{i}\in \mathcal{S}$ assign task $s_{i}$ to worker-interval pair $(j, \ell)$, where the worker-interval pair $(j, \ell)$ is chosen from the probability distribution that assigns task $s_{i}$ to $(j, \ell)$ with probability $\frac{y_{ij\ell}\cdot |I_{\ell}|}{\tau_{i}}$; set $t_{i}$ to the left endpoint of time interval $I_{\ell}$ and $\mathcal{S}_{j}=\mathcal{S}_{j}\cup \left\{s_{i}\right\}$. \label{alg3-3}
				\STATE Schedule on each crowd worker $u_{j}$ the tasks that were assigned to it as early as possible in non-decreasing order of $t_{i}$, breaking ties independently at random. \label{alg3-4}
   \end{algorithmic}
\label{Alg3}
\end{algorithm}

The Randomized Interval-indexed Scheduling (RIS) algorithm is described in Algorithm~\ref{Alg3}. In lines~\ref{alg3-2} to~\ref{alg3-3}, the start time for processing each task is determined randomly, based on the solution to the linear program~(\ref{coflow:interval}). In line~\ref{alg3-4}, tasks are scheduled in non-decreasing order of $t_{i}$, with ties broken independently at random.


The following Theorem analyzes the approximation ratio of Algorithm~\ref{Alg3}.
\begin{thm}\label{thm:thm44}
In the schedule generated by Algorithm~\ref{Alg3}, the expected completion time for each task $s_{i}$ is bounded by $(1.5+\frac{\eta}{2})C_{i}^{*}$ Here, $C_{i}^{*}$ represents the solution derived from the linear programming formulation (\ref{coflow:interval}).
\end{thm}
\begin{proof}
For any task $s_{i} \in \mathcal{S}$, we assume that Algorithm~\ref{Alg3} assigns the worker-interval $(j, \ell)$ to this task. Let $\mathcal{K}$ represent the set of tasks that have been assigned to crowd worker $u_{j}$ and started no later than $s_{i}$ (including $s_{i}$). We have $C_{i} = e_{j}+\sum_{k\in \mathcal{K}} \tau_{k}$. 

To analyze the expected completion time $E[C_{i}]$, we start by deriving an upper bound on the conditional expectation $E_{\ell, j}[C_{i}]$. Consider $\ell=0$, we have
\begin{eqnarray}\label{thm4:eq3}
 E_{\ell=0, j}[C_{i}] & \leq & E_{\ell=0, j}\left[e_{j}+\sum_{k\in \mathcal{K}} \tau_{k}\right] \notag\\
							& \leq & e_{j}+\tau_{i} +\sum_{k\in \mathcal{K}\setminus \left\{i\right\}} \tau_{k} \cdot Pr_{\ell=0, j}(k)  \notag\\
              & \leq & e_{j}+\tau_{i} +\sum_{k\in \mathcal{K}\setminus \left\{i\right\}} \frac{1}{2}y_{kj0}\cdot |I_{0}| \notag\\
							& \leq & e_{j}+\tau_{i} + \frac{1}{2}|I_{0}| \notag\\ 
							& =    & e_{j}+\tau_{i} + \frac{1}{2} \notag\\
							& =    & \left(e_{j}+(1+\eta)^{\ell-1}+\frac{1}{2}\tau_{i}\right)+\frac{1}{2}\tau_{i}
\end{eqnarray}
where $Pr_{\ell=0, j}(k)=\frac{1}{2}\frac{y_{kj\ell}\cdot |I_{\ell}|}{\tau_{k}}$ represents the probability of $s_{k}$ on crowd worker $u_{j}$ prior to the occurrence of $s_{i}$. Note that the factor $\frac{1}{2}$ preceding the term $\frac{y_{kj\ell}\cdot |I_{\ell}|}{\tau_{k}}$ is introduced to account for random tie-breaking. Additionally, remember that the expression $(1+\eta)^{\ell-1}$ is defined as $\frac{1}{2}$ when $\ell = 0$.

When $\ell>0$, we have
\begin{eqnarray}\label{thm4:eq4}
E_{\ell, j}[C_{i}]  & \leq & E_{\ell, p}\left[e_{j}+\sum_{k\in \mathcal{K}} \tau_{k}\right] \notag\\
              & \leq & e_{j}+\tau_{i} +\sum_{k\in \mathcal{K}\setminus \left\{i\right\}} \tau_{k} \cdot Pr_{\ell, j}(k) \notag\\
              & \leq & e_{j}+\tau_{i} + \sum_{t=0}^{\ell-1}|I_{t}|+\frac{1}{2}|I_{\ell}| \notag\\
							& \leq & e_{j}+(1+\frac{\eta}{2})(1+\eta)^{\ell-1}+\tau_{i} \notag\\
							& \leq & \textstyle{(1+\frac{\eta}{2})\left(e_{j}+(1+\eta)^{\ell-1}+\frac{1}{2}\tau_{i}\right)+\frac{1}{2}\tau_{i}}
\end{eqnarray}
where $Pr_{\ell, j}(k)=\sum_{t=0}^{\ell-1}\frac{y_{kjt}\cdot |I_{t}|}{\tau_{k}}+\frac{1}{2}\frac{y_{kj\ell}\cdot |I_{\ell}|}{\tau_{k}}$.

We are now applying the total expectation formula to eliminate the conditioning results in inequalities (\ref{thm4:eq3}) and (\ref{thm4:eq4}). We have
\begin{eqnarray*}\label{thm4:eq5}
E[C_{i}] & =    & \sum_{j=1}^{m}\sum_{\ell=0}^{L} \frac{y_{ij\ell}\cdot |I_{\ell}|}{\tau_{i}} E_{\ell, j}[C_{i}] \\
& \leq & \left(1+\frac{\eta}{2}\right) \sum_{j=1}^{m}\sum_{\ell=0}^{L}\left(\frac{y_{ij\ell}\cdot |I_{\ell}|}{\tau_{i}}\left(e_{j}+(1+\eta)^{\ell-1}\right)+\frac{1}{2}y_{ij\ell}\cdot |I_{\ell}|\right)+\frac{1}{2} \sum_{j=1}^{m}\sum_{\ell=0}^{L}y_{ij\ell}\cdot|I_{\ell}|.
\end{eqnarray*}
This inequality, along with constraints (\ref{interval:d}) and (\ref{interval:e}), leads to the conclusion of the theorem.
\end{proof}

For any given $\epsilon>0$, let $\eta=2\cdot \epsilon$. Then, Algorithm~\ref{Alg3} achieves expected approximation ratios of $1.5+\epsilon$.

\section{Deterministic Approximation Algorithm with Interval-indexed Approach}\label{sec:Algorithm4}
Algorithm~\ref{Alg3}. In lines~\ref{alg4-2} to \ref{alg4-3}, we select a worker-interval pair $(j, \ell)$ for each task to minimize the expected total weighted completion time. In line \ref{alg4-4}, tasks assigned to each crowd worker are scheduled in non-decreasing order based on $t_{i}$, with ties broken by selecting the task with the smaller index.

\begin{algorithm}
\caption{The DIS Algorithm}
    \begin{algorithmic}[1]
				\STATE Compute an optimal solution $y$ to linear programming (\ref{coflow:interval}).
				\STATE Set $\mathcal{P}=\emptyset$; $x=0$. \label{alg4-2}
				\FOR{all $s_{i}\in \mathcal{S}$} 
					\STATE for all possible assignments of $s_{i}\in \mathcal{S}\setminus \mathcal{P}$ to worker-interval pair $(j, \ell)$ compute $E_{\mathcal{P}\cup \left\{s_{i}\right\}, x}[\sum_{q}w_{q} C_{q}]$; 
					\STATE Determine the worker-interval pair $(j, \ell)$ that minimizes the conditional expectation $E_{\mathcal{P}\cup \left\{s_{i}\right\}, x}[\sum_{q}w_{q} C_{q}]$
					\STATE Set $\mathcal{P}=\mathcal{P}\cup \left\{s_{i}\right\}$; $x_{ij\ell}=1$; set $t_{i}$ to the left endpoint of time interval $I_{\ell}$; set $\mathcal{S}_{j}=\mathcal{S}_{j}\cup \left\{s_{i}\right\}$;
				\ENDFOR \label{alg4-3}
				\STATE Schedule on each crowd worker $u_{j}$ the tasks that were assigned to it as early as possible in non-decreasing order of $t_{i}$, breaking ties with smaller indices. \label{alg4-4}
   \end{algorithmic}
\label{Alg4}
\end{algorithm}

To analyze the expected completion time $E[C_{i}]$, we start by deriving an upper bound on the conditional expectation $E_{\ell, j}[C_{i}]$:
\begin{eqnarray*}
E_{\ell, j}[C_{i}] & = & e_{j}+\tau_{i}  +\sum_{k\in \mathcal{S}\setminus \left\{i\right\}} \sum_{t=0}^{\ell-1}y_{kjt}\cdot |I_{t}| +\sum_{k<i} y_{kj\ell} \cdot |I_{\ell}|.
\end{eqnarray*}
The expected completion time of task $s_{i}$ is
\begin{eqnarray*}
E[C_{i}] & = & \sum_{j=1}^{m}\sum_{\ell=0}^{L} \frac{y_{ij\ell}\cdot |I_{\ell}|}{\tau_{i}} E_{\ell, j}[C_{i}].
\end{eqnarray*}

Let $\mathcal{P}\subseteq \mathcal{S}$ represent a subset of tasks that have already been assigned to the time interval on crowd worker $u_{j}$. Define $\mathcal{P}^{+}=\mathcal{P}\cup \left\{s_{i}\right\}$ as the set that includes the task $s_{i}$, and $\mathcal{P}^{-}=\mathcal{P}\setminus \left\{s_{i}\right\}$ as the set that excludes the task $s_{i}$. For each task $s_{i'}\in \mathcal{P}$, let the binary variable $x_{i'j\ell}$ indicate whether the task $s_{i'}$ has been assigned to the time interval $I_{\ell}$ for crowd worker $u_{j}$. Specifically, $x_{i'j\ell}=1$ if $s_{i'}$ is assigned, and $x_{i'j\ell}=0$ if it is not assigned. These notations enable us to formulate expressions for the conditional expectation of the completion time for task $s_{i}$.
Let
\begin{eqnarray*}
D(i,j,\ell)& = & e_{j}+\tau_{i} +\sum_{k\in \mathcal{P}} \sum_{t=0}^{\ell-1}x_{kjt}\tau_{k} +\sum_{\substack{k\in \mathcal{P}\\ k<i}} x_{kj\ell}\tau_{k} \\
					 &   & +\sum_{k\in \mathcal{S}\setminus \mathcal{P}^{+}} \sum_{t=0}^{\ell-1} y_{kjt}\cdot |I_{t}| +\sum_{\substack{k\in \mathcal{S}\setminus \mathcal{P}^{+}\\ k<i}} y_{kj\ell}\cdot |I_{\ell}|
\end{eqnarray*}
be the conditional expectation completion time for task $s_{i}$, which has been assigned to crowd worker $u_{j}$ by $\ell$, when $s_{i}\notin \mathcal{P}$.

Let
\begin{eqnarray*}
E(i,j,\ell)& = & e_{j}+\tau_{i} +\sum_{k\in \mathcal{P}^{-}} \sum_{t=0}^{\ell-1}x_{kjt}\tau_{k} +\sum_{\substack{k\in \mathcal{P}\\ k<i}} x_{kj\ell}\tau_{k} \\
					 &   & +\sum_{k\in \mathcal{K}\setminus \mathcal{P}} \sum_{t=0}^{\ell-1} y_{kjt}\cdot |I_{t}| +\sum_{\substack{k\in \mathcal{K}\setminus \mathcal{P} \\ k<i}} y_{kj\ell}\cdot |I_{\ell}|
\end{eqnarray*}
be the expectation completion time for task $s_{i}$, which has been assigned to crowd worker $u_{j}$ by $\ell$, when $s_{i}\in \mathcal{P}$.

If $s_{i}\in \mathcal{P}$, we have
\begin{eqnarray*}
E_{\mathcal{P},x}[C_{i}] & = & \sum_{j=1}^{m} \sum_{\ell=0}^{L} \frac{y_{ij\ell}\cdot |I_{\ell}|}{\tau_{i}} D(i,j,\ell)
\end{eqnarray*}
and, if $s_{i}\notin \mathcal{P}$, we obtain
\begin{eqnarray*}
E_{\mathcal{P},x}[C_{i}] & = & E(i,j,t)
\end{eqnarray*}
where $(j, t)$ is the worker-interval pair task $s_{i}$ has been assigned to, i.e., $x_{ijt} = 1$.

\begin{lem}\label{lem:lem3}
Let $y$ be an optimal solution to linear programming (\ref{coflow:interval}). Let $x$ represent the tasks already assigned to worker-interval pairs. Furthermore, for each $s_{i}\in \mathcal{K}\setminus \mathcal{P}$, there exists an assignment of $s_{i}$ to time interval $I_{\ell}$ on crowd worker $u_{j}$ such that
\begin{eqnarray*}
E_{\mathcal{P}\cup \left\{s_{i}\right\}, x}\left[\sum_{q}w_{q} C_{q}\right] \leq E_{\mathcal{P}, x}\left[\sum_{q}w_{q} C_{q}\right].
\end{eqnarray*}
\end{lem}
\begin{proof}
The expression $E_{\mathcal{P}, x}[\sum_{q}w_{q} C_{q}]$ can be represented as a convex combination of conditional expectations $E_{\mathcal{P}\cup \left\{s_{i}\right\}, x}[\sum_{q}w_{q} C_{q}]$ for all possible assignments of the task $s_{i}$ to time interval $I_{\ell}$ and crowd worker $u_{j}$. The coefficients for this combination are given by $\frac{y_{ij\ell}\cdot |I_{\ell}|}{\tau_{i}}$. Thus, we can select a decision such that the inequality $E_{\mathcal{P}\cup \left\{s_{i}\right\}, x}\left[\sum_{q}w_{q} C_{q}\right] \leq E_{\mathcal{P}, x}\left[\sum_{q}w_{q} C_{q}\right]$ holds, leading to the claimed result.				
\end{proof}

\begin{thm}\label{thm:thm5}
Algorithm~\ref{Alg4} is a deterministic algorithm that offers a performance guarantee of $\max\left\{2.5,1+\epsilon\right\}$.	
\end{thm}
\begin{proof}
Let $\mathcal{K}$ be the set of tasks assigned to crowd worker $u_{j}$ and started no later than $s_{i}$ (including $s_{i}$). To analyze the expected completion time $E[C_{i}]$, we start by deriving an upper bound on the conditional expectation $E_{\ell, j}[C_{i}]$. Consider $\ell=0$, we have
\begin{eqnarray}\label{thm5:eq3}
E_{\ell=0, j}[C_{i}]  & \leq & e_{j}+E_{\ell=0, j}\left[\sum_{k\in \mathcal{K}} \tau_{k}\right] \notag\\
              & \leq & e_{j}+\tau_{i} +\sum_{k\in \mathcal{K}\setminus \left\{i\right\}} y_{kj0}\notag\\
							& \leq & e_{j}+\tau_{i} + 1 \notag\\
							& =    & 2\left(e_{j}+(1+\eta)^{\ell-1}+\frac{1}{2}\tau_{i}\right).
\end{eqnarray}
When $\ell>0$, we have
\begin{eqnarray*}\label{thm5:eq4}
E_{\ell, j}[C_{i}]  & \leq & e_{j}+E_{\ell, j}\left[\sum_{k\in \mathcal{K}} \tau_{k}\right] \notag\\
              & \leq & e_{j}+\tau_{i}+\sum_{k\in \mathcal{K}\setminus \left\{i\right\}} \tau_{k} \cdot Pr_{\ell,j}(k) \notag\\
              & \leq & e_{j}+\tau_{i}+ \sum_{t=0}^{\ell}|I_{t}| \notag\\
							& \leq & e_{j}+\tau_{i}+ (1+\eta)^{\ell}
\end{eqnarray*}
where $Pr_{\ell,j}(k)=\sum_{t=0}^{\ell}\frac{y_{kjt}\cdot |I_{t}|}{\tau_{k}}$. When $(1+\eta)>2$ and $\ell>0$, we have 
\begin{eqnarray}\label{thm5:eq5}
E_{\ell, j}[C_{i}]\leq \left(1+\eta\right)\left(e_{j}+(1+\eta)^{\ell-1}+\frac{1}{2}\tau_{i}\right).
\end{eqnarray}
When $(1+\eta)\leq 2$ and $\ell>0$, we have 
\begin{eqnarray}\label{thm5:eq6}
E_{\ell, j}[C_{i}]\leq\left(1+\eta\right)\left(e_{j}+(1+\eta)^{\ell-1}+\frac{1}{2}\tau_{i}\right)+\frac{1}{2}\tau_{i}.
\end{eqnarray}

When $(1+\eta)>2$, we apply the formula of total expectation to eliminate conditioning results in inequalities (\ref{thm5:eq3}) and (\ref{thm5:eq5}) which yields
\begin{eqnarray*}
E[C_{i}] & \leq & (1+\eta) \sum_{j=1}^{m}\sum_{\ell=0}^{L}\left(\frac{y_{ij\ell}\cdot |I_{\ell}|}{\tau_{i}}\left(e_{j}+(1+\eta)^{\ell-1}\right)+\frac{1}{2}y_{ij\ell}\cdot |I_{\ell}|\right).
\end{eqnarray*}
When $(1+\eta)\leq 2$, we apply the formula of total expectation to eliminate conditioning results in inequalities (\ref{thm5:eq3}) and (\ref{thm5:eq6}) which yields
\begin{eqnarray*}
E[C_{i}] & \leq & 2\cdot \sum_{j=1}^{m}\sum_{\ell=0}^{L}\left(\frac{y_{ij\ell}\cdot |I_{\ell}|}{\tau_{i}}\left(e_{j}+(1+\eta)^{\ell-1}\right)+\frac{1}{2}y_{ij\ell}\cdot |I_{\ell}|\right) +\frac{1}{2} \sum_{j=1}^{m}\sum_{\ell=0}^{L}y_{ij\ell}\cdot|I_{\ell}|.
\end{eqnarray*}

When $\eta=\epsilon$ and $\epsilon>0$, this together with the constraints (\ref{interval:d}) and (\ref{interval:e}), along with the inductive application of Lemma~\ref{lem:lem3}, leads to the theorem.
\end{proof}

In Theorem~\ref{thm:thm5}, the upper bound of $E_{\ell=0, j}[C_{i}]$ is overestimated. If we can ensure that the required service time for the task or the total contact time is sufficiently large, we can derive the following theorem.
\begin{thm}\label{thm:thm7}
If $2\cdot e_{min}+\tau_{min}\geq \frac{1-\eta}{\eta}$, Algorithm~\ref{Alg4} is a deterministic algorithm that offers a performance guarantee of $1.5+\epsilon$.	
\end{thm}
\begin{proof}
If $2\cdot e_{min}+\tau_{min}\geq \frac{1-\eta}{\eta}$, then
\begin{eqnarray}\label{thm7:eq1}
E_{\ell=0, j}[C_{i}]  & \leq &  e_{j}+\tau_{i} + 1 \notag\\
							        & \leq & (1+\eta) \left(e_{j}+(1+\eta)^{\ell-1}+\frac{1}{2}\tau_{i}\right) +\frac{1}{2}\tau_{i}.
\end{eqnarray}
When $\ell>0$, we have
\begin{eqnarray}\label{thm7:eq2}
E_{\ell, j}[C_{i}]  & \leq & e_{j}+\tau_{i}+ (1+\eta)^{\ell} \notag\\
							 & \leq & \left(1+\eta\right)\left(e_{j}+(1+\eta)^{\ell-1}+\frac{1}{2}\tau_{i}\right) +\frac{1}{2}\tau_{i}.
\end{eqnarray}
We apply the formula of total expectation to eliminate conditioning results in inequalities (\ref{thm7:eq1}) and (\ref{thm7:eq2}) which yields
\begin{eqnarray*}
 E[C_{i}] & \leq & \left(1+\eta\right)\cdot \sum_{j=1}^{m}\sum_{\ell=0}^{L}\left(\frac{y_{ij\ell}\cdot |I_{\ell}|}{\tau_{i}}\left(e_{j}+(1+\eta)^{\ell-1}\right)+\frac{1}{2}y_{ij\ell}\cdot |I_{\ell}|\right) +\frac{1}{2} \sum_{j=1}^{m}\sum_{\ell=0}^{L}y_{ij\ell}\cdot|I_{\ell}|.
\end{eqnarray*}
This together with the constraints (\ref{interval:d}) and (\ref{interval:e}), along with the inductive application of Lemma~\ref{lem:lem3}, leads to the theorem.
\end{proof}

\subsection{Online Task Scheduling}
This section presents the Online Deterministic Interval-indexed Scheduling (ODIS) algorithm, as shown in Algorithm~\ref{Alg5}. This algorithm modifies Algorithm~\ref{Alg_ONLRF} by adopting the task assignment method from Algorithm~\ref{Alg4} to allocate tasks to crowd workers.
\begin{algorithm}
\caption{The ODIS Algorithm}
    \begin{algorithmic}[1]
		    \STATE $\mathcal{S}=\left\{s_{1},s_{2}, \ldots, s_{n}\right\}$
				\STATE $\mathcal{U}=\left\{u_{1},u_{2}, \ldots, u_{m}\right\}$
				\WHILE{$u_{0}$ meets $u_{j}$}
					\STATE $e_{j}=\frac{1}{\lambda_{j}}$
					\FOR{each $k$ such that $u_{k}\in \mathcal{U}\setminus\left\{u_{j}\right\}$}
						\STATE $e_{k}=\frac{2}{\lambda_{k}}-\frac{1}{\lambda_{j}}$
					\ENDFOR
					\STATE Apply algorithm~\ref{Alg4} to allocate the set of tasks $\mathcal{S}_{j}$ to machine $u_{j}\in \mathcal{U}$.
					\STATE Assign $\mathcal{S}_{j}$ to $u_{j}$ and schedule its execution in non-decreasing order of $t_{i}$ for $s_{i}\in \mathcal{S}_{j}$, breaking ties with smaller indices.
					\STATE $\mathcal{S}=\mathcal{S}\setminus \mathcal{S}_{j}$
					\STATE $\mathcal{U}=\mathcal{U}\setminus \left\{u_{j}\right\}$
				\ENDWHILE
				\STATE \textbf{return} $\Lambda_{CosMOS}=\left\{\mathcal{S}_{1}, \mathcal{S}_{2}, \ldots, \mathcal{S}_{m}\right\}$
   \end{algorithmic}
\label{Alg5}
\end{algorithm}

The following theorem analyzes the competitive ratio of Algorithm~\ref{Alg5}.

\begin{thm}\label{thm:thm6}
Suppose an entity with full knowledge of the mobility of all crowd workers exists, specifically knowing the exact times of each meeting between the requester and the crowd workers in advance. Using this information, the entity can make optimal online task scheduling decisions, represented as $\Lambda_{OPT}=\left\{\mathcal{S}_{1}^{*}, \mathcal{S}_{2}^{*}, \ldots, \mathcal{S}_{m}^{*}\right\}$. Then we have 
\begin{eqnarray*}
\frac{WCT_{ODIS}}{WCT_{OPT}} \leq \alpha\left(1 + \frac{n\cdot w_{max}\cdot\frac{2}{\lambda_{min}}}{w_{min}\sum_{i=1}^{n}\tau_{i}}\right)
\end{eqnarray*} 
where $\alpha =\max\left\{2,1+\epsilon\right\}+0.5$ and $\alpha =1.5+\epsilon$ if $2\cdot e_{min}+\tau_{min}\geq \frac{1-\eta}{\eta}$.
\end{thm}
\begin{proof}
The proof follows a similar approach to that of Theorem~\ref{thm:thm3}.
\end{proof}

\section{Results and Discussion}\label{sec:Results}
This section conducts simulations to evaluate the performance of the proposed algorithm against a previous one, utilizing both synthetic and real datasets. The results of the simulations are presented and analyzed in the following sections.

\subsection{Synthetic Datasets and Settings}
Let $p$ be the probability that the task ratio $\frac{w}{\tau}$ equals 1.
In our synthetic dataset, we assess the impact of various factors, including the probability $p$, the number of crowd workers, the ratio of crowd workers to tasks $\frac{n}{m}$, as well as the mean and standard deviation of the service time required for tasks. 
Below are the default values for these parameters:
\begin{itemize}
\item For each crowd worker $u_{j}$, the service rate $\lambda_{j}$ is randomly assigned from a uniform distribution within the range of $[1, 30]$.
\item The total number of crowd workers is set to 10.
\item The ratio $\frac{n}{m}$ is established at 5.
\item The required service time for each task is generated as a positive random number following a Gaussian distribution, with the mean and variance set to 30 time units.
\end{itemize}
To evaluate the algorithm's performance under worst-case conditions, we ensure that the weight of each task is equal to its required service time.

\subsection{Real Datasets and Settings}
The real datasets are from the Cambridge Haggle datasets~\cite{c70011-22}, which contain records of Bluetooth device connections made by individuals carrying mobile devices over several days. These datasets encompass various scenarios, including office environments, conference settings, and urban areas. With the precise timestamps of each connection event, we can estimate the parameters of the exponential distribution. In our analysis, we selected the mobile devices known as iMotes as the requesters, while the external devices they connected to were considered crowd workers. Since multiple requesters are present in each dataset, we averaged their results to establish a final evaluation criterion.

For each crowd worker $u_{j}$, we define $\lambda_{j}$ as follows: 
\begin{eqnarray*}
\lambda_{j} = \frac{l_{j}}{\sum_{i=1}^{l_{j}}\Delta t_{i,j}}.
\end{eqnarray*}
In this equation, $\Delta t_{i,j}$ represents the inter-meeting time between the $(i-1)$-th and $i$-th contacts between $u_{0}$ and $u_{j}$, while $l_{j}$ is the total number of contacts between $u_{0}$ and $u_{j}$. We select the external devices with the 128 highest $\lambda_{j}$ values as the crowd workers for each dataset. Additionally, we determine the ratio of the number of tasks to the number of crowd workers, denoted as $\frac{n}{m}$, from the set $\left\{2, 4, 6, 8, 10\right\}$.

The required service time for each task is modeled as a Gaussian-distributed random positive number, with both the mean and the variance set to 30 time units. To assess the algorithm's performance under worst-case scenarios, we ensure that the weight of each task equals its required service time.


\subsection{Algorithms}
Due to the time-consuming nature of the derandomized method of the DIS Algorithm, we propose the Modified DIS (MDIS) Algorithm (Algorithm~\ref{Alg6}) in this evaluation. This algorithm directly assigns the solutions of the linear program $\bar{C}_{i}$ in non-decreasing order to the least loaded crowd worker. We evaluate the performance of the LRF Algorithm (Algorithm~\ref{Alg_LRF}), the RIS Algorithm (Algorithm~\ref{Alg3}), and the MDIS Algorithm (Algorithm~\ref{Alg6}). In our simulations, we compute the approximation ratio by dividing the total weighted completion time obtained from the algorithms by the solution of linear programming (\ref{coflow:interval}). We generated 100 instances for each case and calculated the average performance of the algorithms.

\begin{algorithm}
\caption{The MDIS Algorithm}
    \begin{algorithmic}[1]
		    \STATE $\mathcal{S}=\left\{s_{1},s_{2}, \ldots, s_{n}\right\}$
				\STATE $\mathcal{U}=\left\{u_{1},u_{2}, \ldots, u_{m}\right\}$
				\STATE Compute an optimal solution $\bar{C}_{i}$ for $s_{i}\in \mathcal{S}$ to linear programming (\ref{coflow:interval}).
				\FOR{$j=1, 2, \ldots, m$} \label{Alg6:1}
					\STATE $\mathcal{S}_{j}=\emptyset$
					\STATE $EW_{j}=\frac{2}{\lambda_{j}}$
				\ENDFOR                    \label{Alg6:2}
				\FOR{every task $s_{i}\in \mathcal{S}$ in non-decreasing order of $\bar{C}_{i}$, breaking ties arbitrarily} \label{Alg6:3}
				  \STATE $j_{min}=\arg\min\left\{EW_{k}|u_{k}\in U\right\}$
					\STATE $\mathcal{S}_{j_{min}}=\mathcal{S}_{j_{min}} \cup \left\{s_{i}\right\}$
					\STATE $EW_{j_{min}}=EW_{j_{min}}+\tau_{i}$
				\ENDFOR \label{Alg6:4}
				\STATE \textbf{return} $\Lambda_{LRF}=\left\{\mathcal{S}_{1}, \mathcal{S}_{2}, \ldots, \mathcal{S}_{m}\right\}$
   \end{algorithmic}
\label{Alg6}
\end{algorithm}

\subsection{Results}
Figure~\ref{fig:ratio1} presents the results for various probability values. The probability $p$ ranges from 0.2 to 1 in increments of 0.2. The weight of each task is randomly selected from the range $[1, 10]$, and with a probability $p$, the weight equals the required service time. Since the solution of the RIS Algorithm is determined randomly, its performance is the worst. 

When $p = 0.2$, the performance of the MDIS Algorithm is $0.11\%$ worse than that of the LRF Algorithm. However, when $p = 0.4$, the MDIS Algorithm outperforms the LRF Algorithm, and as $p$ continues to increase,  the performance advantage of the MDIS Algorithm also rises, ranging from $0.01\%$ to $0.35\%$. This improvement occurs because as more tasks have the same task ratio $\frac{w}{\tau}$, the LRF Algorithm struggles to determine their scheduling order effectively, resulting in reduced performance.

\begin{figure}[!ht]
    \centering
        \includegraphics[width=3.8in]{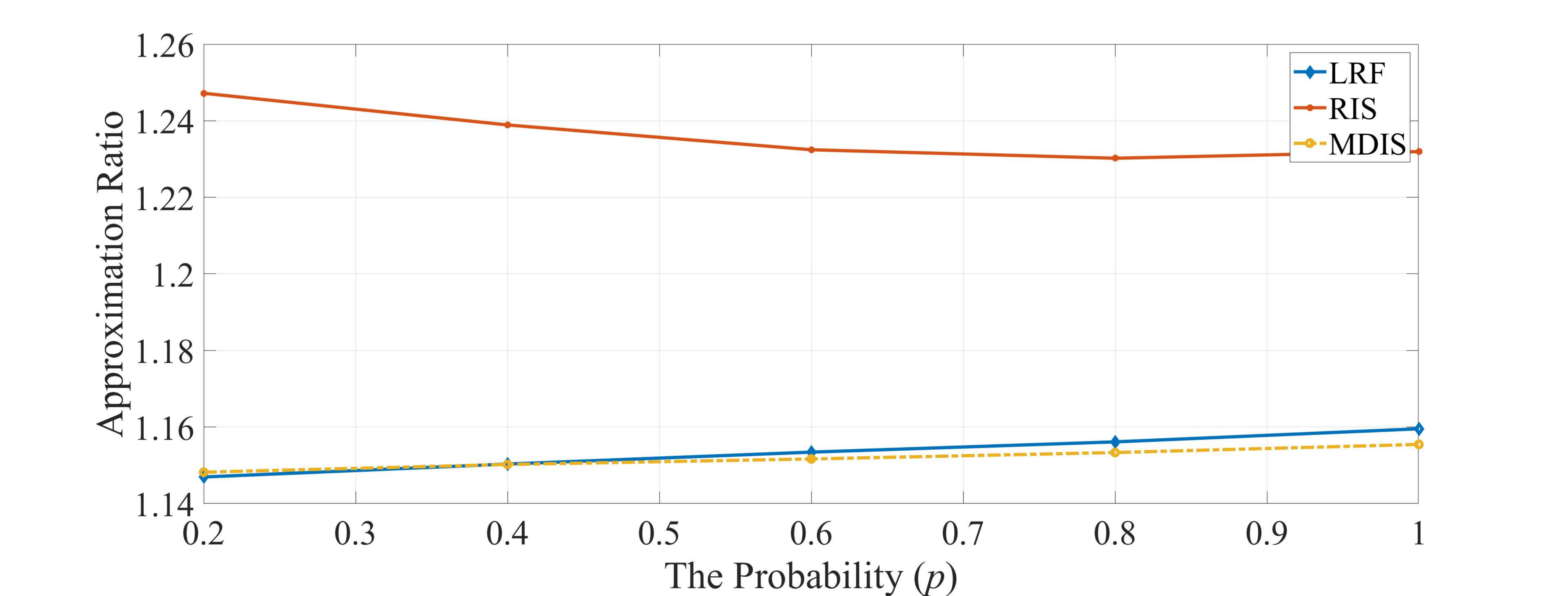}
    \caption{The performance of algorithms for various probability values.}
    \label{fig:ratio1}
\end{figure}

Figure~\ref{fig:ratio2} illustrates the results for different numbers of tasks. The ratio $\frac{n}{m}$ varies from 1 to 10 in increments of 1. Among the algorithms, the MDIS algorithm performs best, followed by the LRF and RIS algorithms. As the number of tasks increases, the performance of all three algorithms gradually converges. This trend occurs because, with more tasks, the influence of the order in which individual tasks are scheduled on the overall total weighted completion time diminishes.

\begin{figure}[!ht]
    \centering
        \includegraphics[width=3.8in]{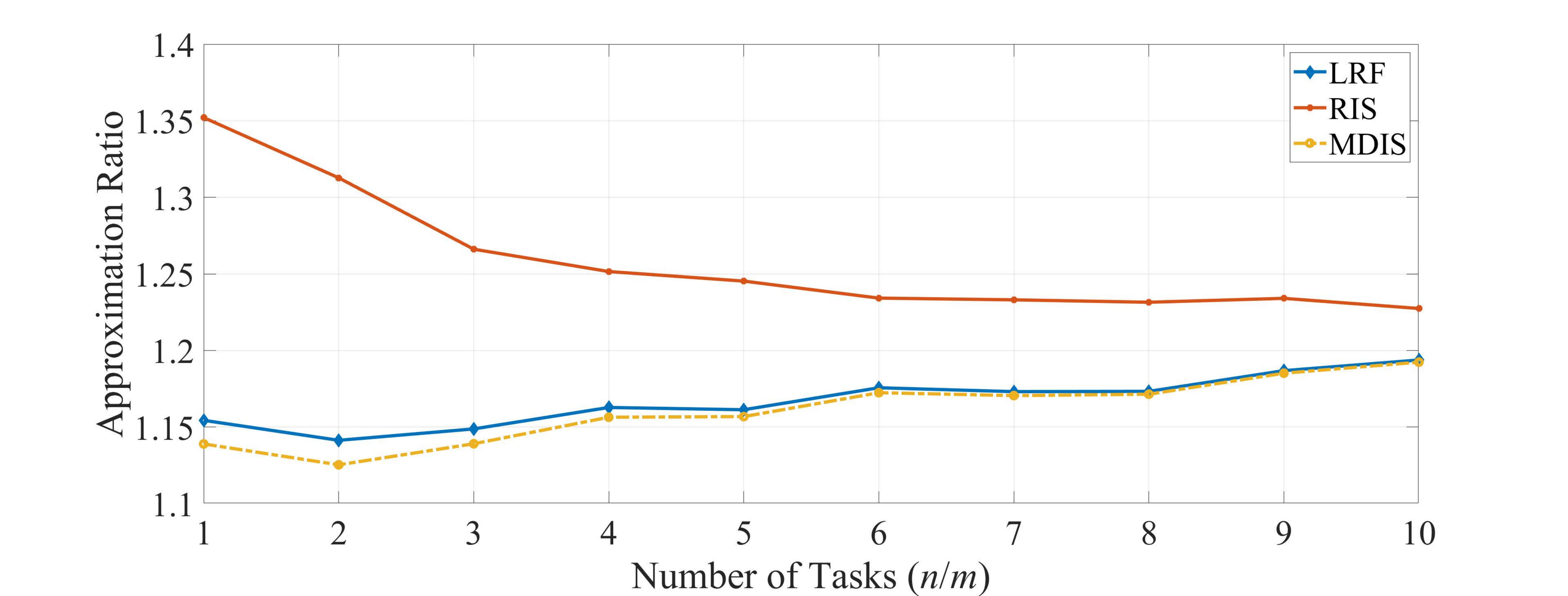}
    \caption{The performance of algorithms with different numbers of tasks.}
    \label{fig:ratio2}
\end{figure}

Figure~\ref{fig:ratio3} shows the results for varying numbers of crowd workers, with the number of workers $m$ ranging from 5 to 25 in increments of 5. The MDIS Algorithm outperforms the other algorithms, followed by the LRF Algorithm and then the RIS Algorithm. When $m = 5$, the MDIS Algorithm's performance is $0.57\%$ better than that of the LRF Algorithm. However, as $m$ increases, the performance advantage of the MDIS Algorithm over the LRF Algorithm declines, decreasing from $0.57\%$ to $0.22\%$. This trend occurs because, as the number of workers increases, the load differences among the crowd workers become more easily balanced.

\begin{figure}[!ht]
    \centering
        \includegraphics[width=3.8in]{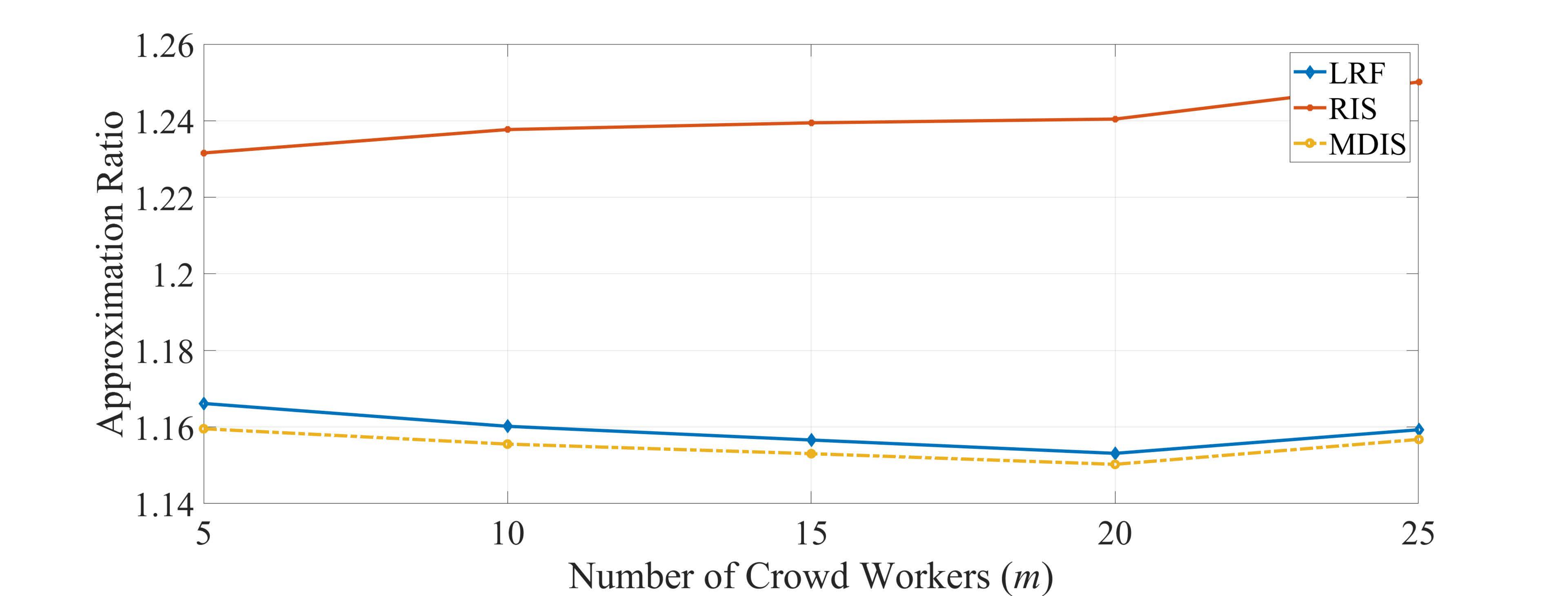}
    \caption{The performance of algorithms for varying numbers of crowd workers.}
    \label{fig:ratio3}
\end{figure}

Figure~\ref{fig:ratio4} displays the results for the mean service time required for tasks, which range from 5 to 50 in increments of 5. The MDIS Algorithm outperforms the others, followed by the LRF and RIS Algorithms. When the average service time is 5, the performance of the MDIS Algorithm is $0.85\%$ better than that of the LRF Algorithm. As the mean service time increases, the performance improvement of the MDIS Algorithm over the LRF Algorithm decreases from $0.85\%$ to $0.28\%$.

\begin{figure}[!ht]
    \centering
        \includegraphics[width=3.8in]{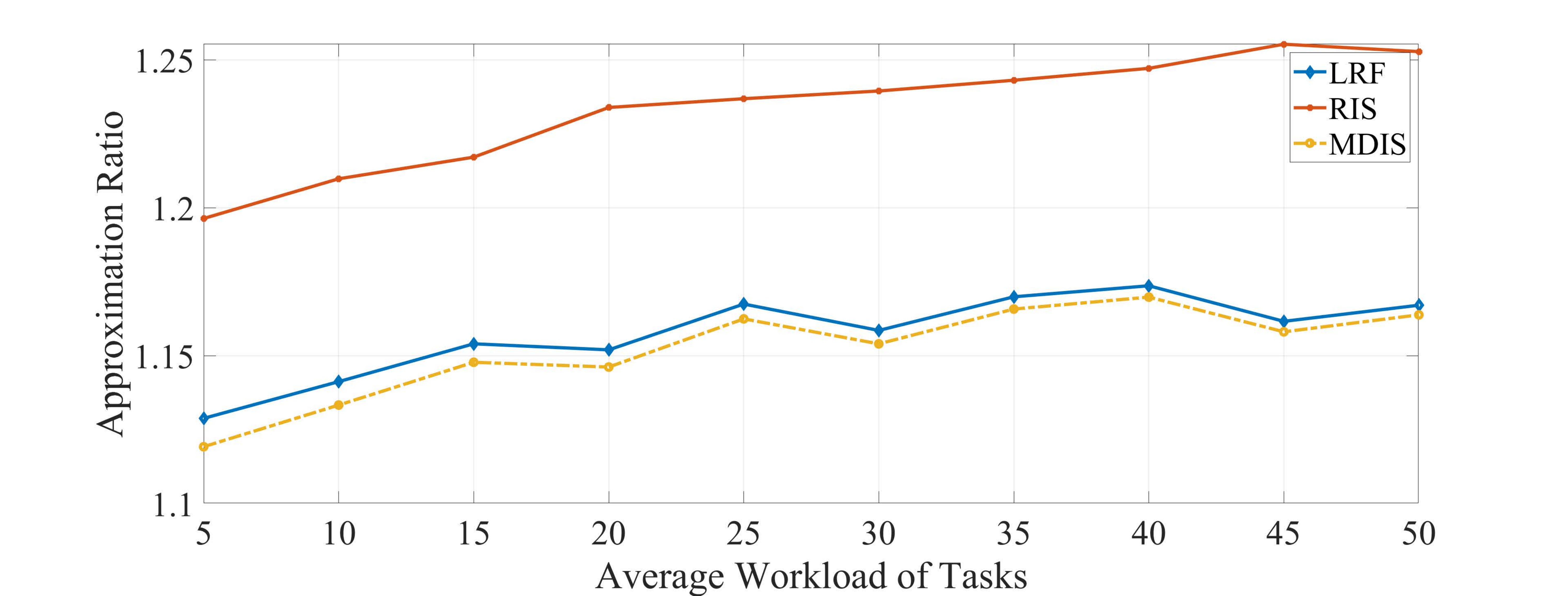}
    \caption{The performance of algorithms with different means of the service time required for tasks.}
    \label{fig:ratio4}
\end{figure}

Figure~\ref{fig:ratio5} illustrates the results for various standard deviations of the service time required for tasks, ranging from 20 to 38 with increments of 2. The MDIS Algorithm demonstrates the best performance, followed by the LRF and RIS Algorithms. When the standard deviation of the service time is 20, the performance of the MDIS Algorithm is $0.37\%$  better than that of the LRF Algorithm. As the mean service time required for tasks increases, the performance improvement of the MDIS Algorithm over the LRF Algorithm rises from $0.37\%$ to $0.45\%$.

\begin{figure}[!ht]
    \centering
        \includegraphics[width=3.8in]{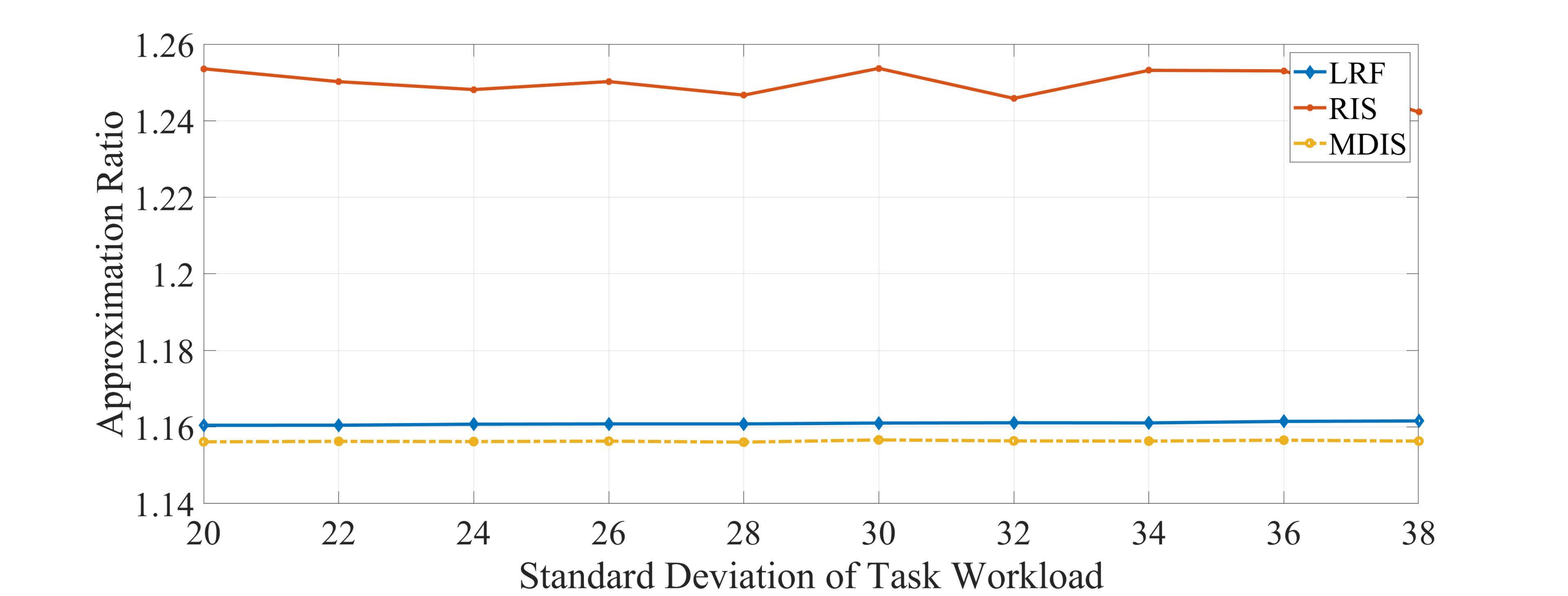}
    \caption{The performance of algorithms for various standard deviations of the service time required for tasks.}
    \label{fig:ratio5}
\end{figure}

Figures \ref{fig:ratio6}, \ref{fig:ratio7}, and \ref{fig:ratio8} display the results with varying numbers of tasks for the real datasets: Intel, Cambridge, and Infocom. The ratio $\frac{n}{m}$ ranges from 2 to 10 in increments of 2. These findings align with the results in Figure \ref{fig:ratio2}. The MDIS algorithm demonstrates the best performance, followed by the LRF and RIS algorithms. As the number of tasks increases, the performance of all three algorithms gradually converges.

\begin{figure}[!ht]
    \centering
        \includegraphics[width=3.8in]{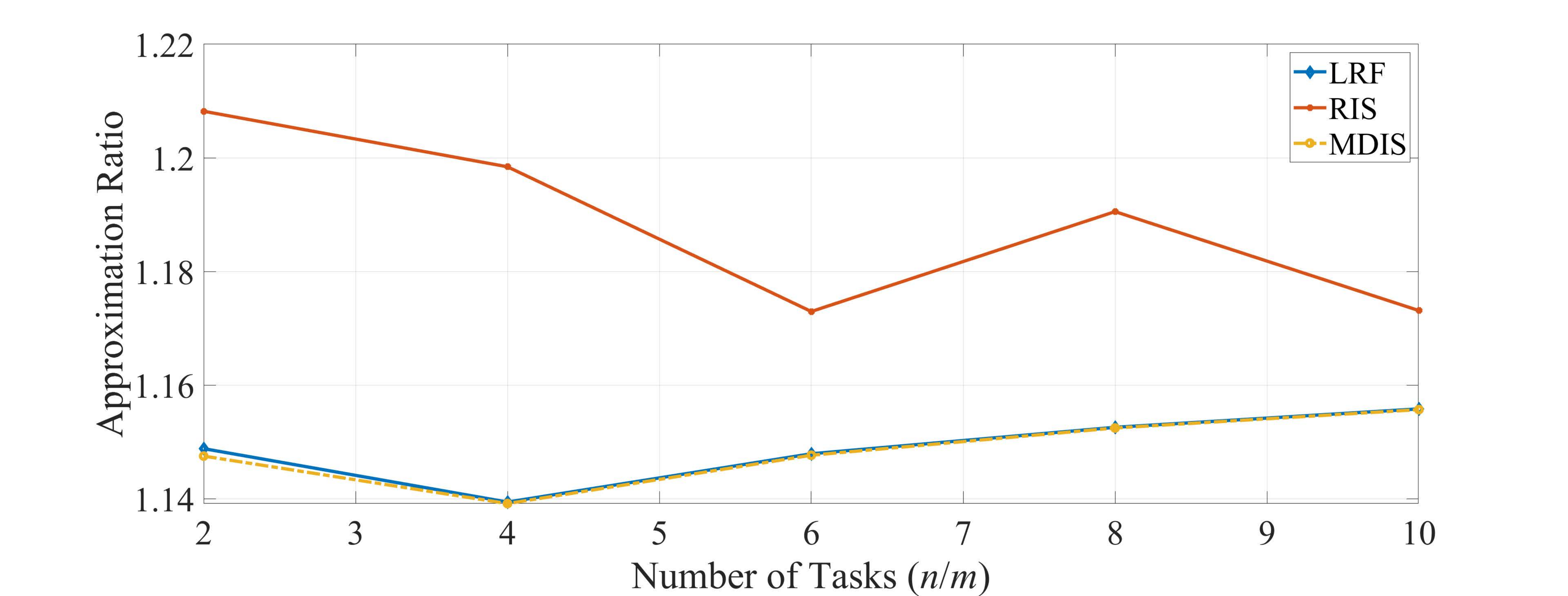}
    \caption{The performance of algorithms with varying numbers of tasks using the Intel real dataset.}
    \label{fig:ratio6}
\end{figure}

\begin{figure}[!ht]
    \centering
        \includegraphics[width=3.8in]{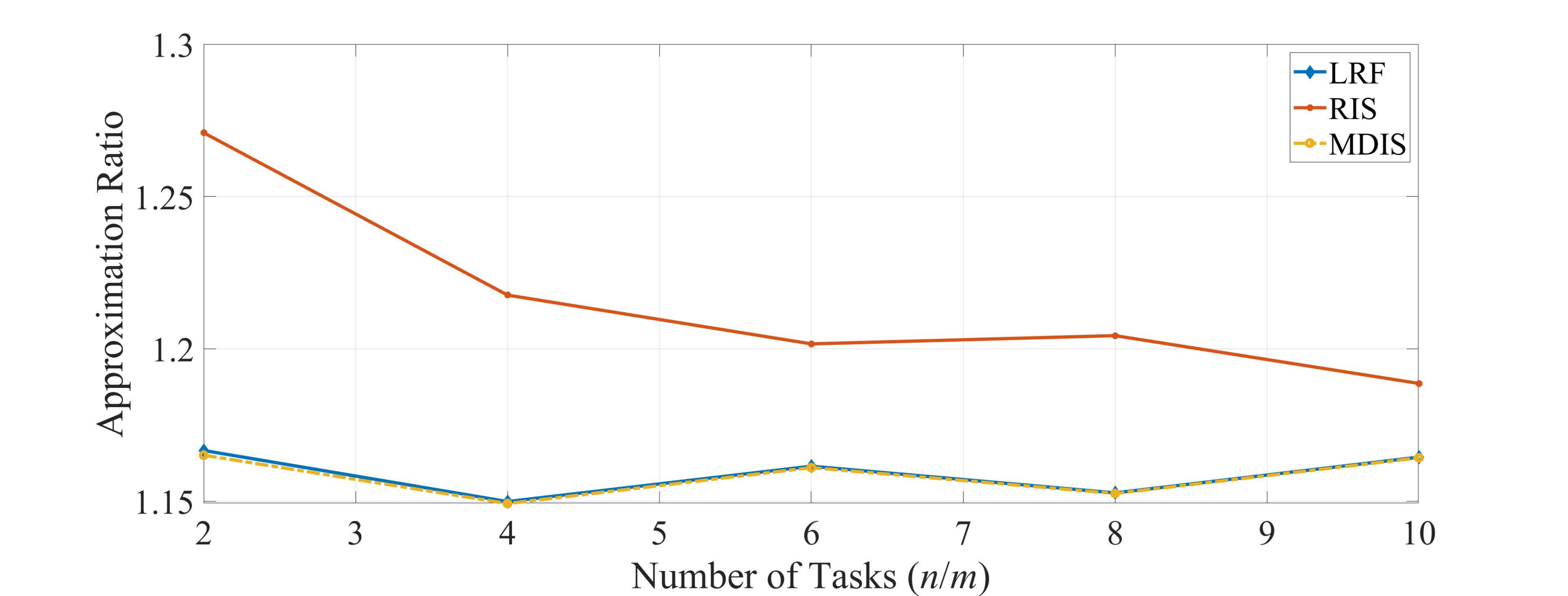}
    \caption{The performance of algorithms with varying numbers of tasks using the Cambridge real dataset.}
    \label{fig:ratio7}
\end{figure}

\begin{figure}[!ht]
    \centering
        \includegraphics[width=3.8in]{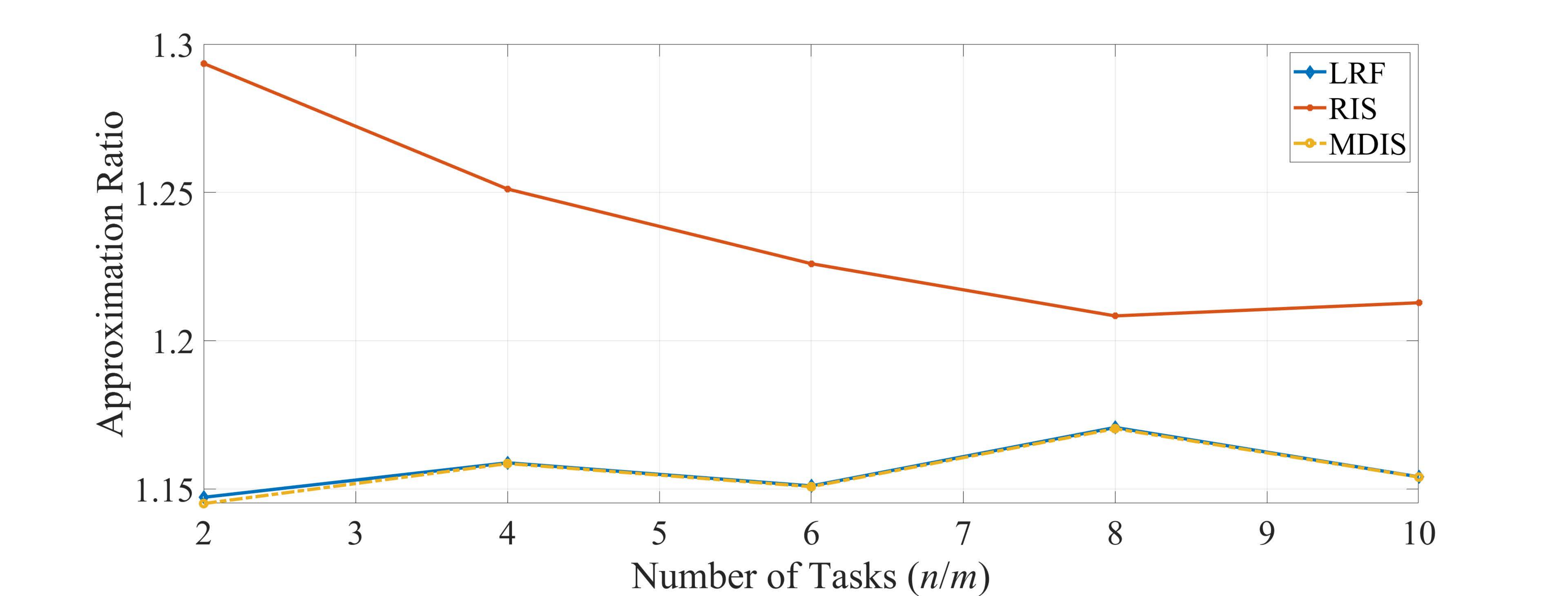}
    \caption{The performance of algorithms with varying numbers of tasks using the Infocom real dataset.}
    \label{fig:ratio8}
\end{figure}

\section{Concluding Remarks}\label{sec:Conclusion}
Conducting mobile crowdsourcing activities within a large-scale system often requires significant resources for management and maintenance. Suppose a requester can distribute crowdsourcing tasks within their private social circle. In that case, they can leverage existing mobile social networks to carry out these tasks, effectively reducing system overhead and saving costs. Therefore, well-designed algorithms can help users utilize the resources of mobile social networks more efficiently. 

First, we demonstrate that the approximation ratio analysis in the paper by Zhang \textit{et al.}~\cite{Zhang2025} is incorrect and provide the correct analysis results. We also prove that when the required service time of a task exceeds the total contact time between the requester and the crowd worker, the approximation ratio of the Largest-Ratio-First based task scheduling algorithm can reach $2 - \frac{1}{m}$. These findings update the conclusions regarding the online algorithm presented in the same paper~\cite{Zhang2025}.

Next, we introduce a randomized approximation algorithm to minimize total weighted completion time in mobile social networks. This algorithm achieves an expected approximation ratio of $1.5 + \epsilon$ for $\epsilon>0$.

Finally, we present a deterministic approximation algorithm to minimize mobile social networks' total weighted completion time. This algorithm has an approximation ratio of $\max\left\{2.5,1+\epsilon\right\}$ for $\epsilon>0$. Additionally, when the task's required service time or the total contact time between the requester and the crowd worker is sufficiently large, this algorithm achieves an approximation ratio of $1.5+\epsilon$ for $\epsilon>0$.

Future work for this paper includes exploring algorithms with lower time complexity. While the proposed deterministic interval-indexed scheduling algorithm offers a better approximation ratio guarantee, it is computationally expensive due to the reliance on an interval-indexed linear program and a derandomized method. 

Another avenue for future research is to examine whether the Largest-Ratio-First based task scheduling algorithm can achieve a better approximation ratio. Our analysis indicates that the solution provided by this algorithm is significantly overestimated, while the optimal solution is severely underestimated. This discrepancy results in the term $\frac{w_{max} \cdot \lambda_{max}}{w_{min} \cdot \lambda_{min}}$ being present. Therefore, developing more accurate estimation techniques may lead to tighter guarantees on the approximation ratio.


\begin{thebibliography}{10}
\providecommand{\url}[1]{#1}
\csname url@rmstyle\endcsname
\providecommand{\newblock}{\relax}
\providecommand{\bibinfo}[2]{#2}
\providecommand\BIBentrySTDinterwordspacing{\spaceskip=0pt\relax}
\providecommand\BIBentryALTinterwordstretchfactor{4}
\providecommand\BIBentryALTinterwordspacing{\spaceskip=\fontdimen2\font plus
\BIBentryALTinterwordstretchfactor\fontdimen3\font minus
  \fontdimen4\font\relax}
\providecommand\BIBforeignlanguage[2]{{%
\expandafter\ifx\csname l@#1\endcsname\relax
\typeout{** WARNING: IEEEtran.bst: No hyphenation pattern has been}%
\typeout{** loaded for the language `#1'. Using the pattern for}%
\typeout{** the default language instead.}%
\else
\language=\csname l@#1\endcsname
\fi
#2}}

\bibitem{Allahverdi2008}
A.~Allahverdi, C.~Ng, T.~Cheng, and M.~Y. Kovalyov, ``A survey of scheduling
  problems with setup times or costs,'' \emph{European Journal of Operational
  Research}, vol. 187, no.~3, pp. 985--1032, 2008.


\bibitem{An2015}
J.~An, X.~Gui, Z.~Wang, J.~Yang, and X.~He, ``A crowdsourcing assignment model
  based on mobile crowd sensing in the internet of things,'' \emph{IEEE
  Internet of Things Journal}, vol.~2, no.~5, pp. 358--369, 2015.


\bibitem{Bhatti2021}
S.~S. Bhatti, J.~Fan, K.~Wang, X.~Gao, F.~Wu, and G.~Chen, ``An approximation
  algorithm for bounded task assignment problem in spatial crowdsourcing,''
  \emph{IEEE Transactions on Mobile Computing}, vol.~20, no.~8, pp. 2536--2549,
  2021.


\bibitem{Bridi2016}
T.~Bridi, A.~Bartolini, M.~Lombardi, M.~Milano, and L.~Benini, ``A constraint
  programming scheduler for heterogeneous high-performance computing
  machines,'' \emph{IEEE Transactions on Parallel and Distributed Systems},
  vol.~27, no.~10, pp. 2781--2794, 2016.


\bibitem{Cheng2017}
Y.~Cheng, Y.~Yuan, L.~Chen, C.~Giraud-Carrier, and G.~Wang, ``Complex
  event-participant planning and its incremental variant,'' in \emph{2017 IEEE
  33rd International Conference on Data Engineering (ICDE)}, 2017, pp.
  859--870.


\bibitem{Cheng2021}
Y.~Cheng, Y.~Yuan, L.~Chen, C.~Giraud-Carrier, G.~Wang, and B.~Li,
  ``Event-participant and incremental planning over event-based social
  networks,'' \emph{IEEE Transactions on Knowledge and Data Engineering},
  vol.~33, no.~2, pp. 474--488, 2021.


\bibitem{Ding2016}
J.-Y. Ding, S.~Song, R.~Zhang, R.~Chiong, and C.~Wu, ``Parallel machine
  scheduling under time-of-use electricity prices: New models and optimization
  approaches,'' \emph{IEEE Transactions on Automation Science and Engineering},
  vol.~13, no.~2, pp. 1138--1154, 2016.


\bibitem{Dutta2009}
P.~Dutta, P.~M. Aoki, N.~Kumar, A.~Mainwaring, C.~Myers, W.~Willett, and
  A.~Woodruff, ``Common sense: participatory urban sensing using a network of
  handheld air quality monitors,'' in \emph{Proceedings of the 7th ACM
  Conference on Embedded Networked Sensor Systems}, ser. SenSys '09.\hskip 1em
  plus 0.5em minus 0.4em\relax New York, NY, USA: Association for Computing
  Machinery, 2009, p. 349–350.


\bibitem{eastman1964}
W.~L. Eastman, S.~Even, and I.~M. Isaacs, ``Bounds for the optimal scheduling
  of n jobs on m processors,'' \emph{Management science}, vol.~11, no.~2, pp.
  268--279, 1964.


\bibitem{Fan2015}
Y.~Fan, H.~Sun, and X.~Liu, ``Poster: Trim: A truthful incentive mechanism for
  dynamic and heterogeneous tasks in mobile crowdsensing,'' in
  \emph{Proceedings of the 21st Annual International Conference on Mobile
  Computing and Networking}, ser. MobiCom '15.\hskip 1em plus 0.5em minus
  0.4em\relax New York, NY, USA: Association for Computing Machinery, 2015, p.
  272–274.


\bibitem{Farkas2014}
K.~Fárkas, d.~Z. Nagy, T.~Tomás, and R.~Szabó, ``Participatory sensing based
  real-time public transport information service,'' in \emph{2014 IEEE
  International Conference on Pervasive Computing and Communication Workshops
  (PERCOM WORKSHOPS)}, 2014, pp. 141--144.


\bibitem{Gao2009}
W.~Gao, Q.~Li, B.~Zhao, and G.~Cao, ``Multicasting in delay tolerant networks:
  a social network perspective,'' in \emph{Proceedings of the Tenth ACM
  International Symposium on Mobile Ad Hoc Networking and Computing}, ser.
  MobiHoc '09.\hskip 1em plus 0.5em minus 0.4em\relax New York, NY, USA:
  Association for Computing Machinery, 2009, p. 299–308.


\bibitem{graham1969bounds}
R.~L. Graham, ``Bounds on multiprocessing timing anomalies,'' \emph{SIAM
  journal on Applied Mathematics}, vol.~17, no.~2, pp. 416--429, 1969.


\bibitem{Guo2017}
B.~Guo, Y.~Liu, W.~Wu, Z.~Yu, and Q.~Han, ``Activecrowd: A framework for
  optimized multitask allocation in mobile crowdsensing systems,'' \emph{IEEE
  Transactions on Human-Machine Systems}, vol.~47, no.~3, pp. 392--403, 2017.


\bibitem{Kawaguchi1986}
T.~Kawaguchi and S.~Kyan, ``Worst case bound of an lrf schedule for the mean
  weighted flow-time problem,'' \emph{SIAM Journal on Computing}, vol.~15,
  no.~4, pp. 1119--1129, 1986.


\bibitem{Liu2012}
X.~Liu, Q.~He, Y.~Tian, W.-C. Lee, J.~McPherson, and J.~Han, ``Event-based
  social networks: linking the online and offline social worlds,'' in
  \emph{Proceedings of the 18th ACM SIGKDD International Conference on
  Knowledge Discovery and Data Mining}, ser. KDD '12.\hskip 1em plus 0.5em
  minus 0.4em\relax New York, NY, USA: Association for Computing Machinery,
  2012, p. 1032–1040.


\bibitem{Maiti2020}
B.~Maiti, R.~Rajaraman, D.~Stalfa, Z.~Svitkina, and A.~Vijayaraghavan,
  ``Scheduling precedence-constrained jobs on related machines with
  communication delay,'' in \emph{2020 IEEE 61st Annual Symposium on
  Foundations of Computer Science (FOCS)}, 2020, pp. 834--845.


\bibitem{Rana2010}
R.~K. Rana, C.~T. Chou, S.~S. Kanhere, N.~Bulusu, and W.~Hu, ``Ear-phone: an
  end-to-end participatory urban noise mapping system,'' in \emph{Proceedings
  of the 9th ACM/IEEE International Conference on Information Processing in
  Sensor Networks}, ser. IPSN '10.\hskip 1em plus 0.5em minus 0.4em\relax New
  York, NY, USA: Association for Computing Machinery, 2010, p. 105–116.


\bibitem{Schulz2002}
A.~S. Schulz and M.~Skutella, ``Scheduling unrelated machines by randomized
  rounding,'' \emph{SIAM Journal on Discrete Mathematics}, vol.~15, no.~4, pp.
  450--469, 2002.


\bibitem{c70011-22}
J.~Scott, R.~Gass, J.~Crowcroft, P.~Hui, C.~Diot, and A.~Chaintreau, ``Crawdad
  cambridge/haggle (v. 2009-05-29),'' 2022.


\bibitem{She2015}
J.~She, Y.~Tong, and L.~Chen, ``Utility-aware social event-participant
  planning,'' in \emph{Proceedings of the 2015 ACM SIGMOD International
  Conference on Management of Data}, ser. SIGMOD '15.\hskip 1em plus 0.5em
  minus 0.4em\relax New York, NY, USA: Association for Computing Machinery,
  2015, p. 1629–1643.


\bibitem{Singh2017}
G.~Singh, D.~Bansal, S.~Sofat, and N.~Aggarwal, ``Smart patrolling: An
  efficient road surface monitoring using smartphone sensors and
  crowdsourcing,'' \emph{Pervasive and Mobile Computing}, vol.~40, pp. 71--88,
  2017.


\bibitem{smith1956}
W.~E. Smith, ``Various optimizers for single-stage production,'' \emph{Naval
  Research Logistics Quarterly}, vol.~3, no. 1-2, pp. 59--66, 1956.


\bibitem{Wu2014}
S.~Wu, X.~Wang, S.~Wang, Z.~Zhang, and A.~K.~H. Tung, ``K-anonymity for
  crowdsourcing database,'' \emph{IEEE Transactions on Knowledge and Data
  Engineering}, vol.~26, no.~9, pp. 2207--2221, 2014.


\bibitem{Xiao2015}
M.~Xiao, J.~Wu, L.~Huang, Y.~Wang, and C.~Liu, ``Multi-task assignment for
  crowdsensing in mobile social networks,'' in \emph{2015 IEEE Conference on
  Computer Communications (INFOCOM)}, 2015, pp. 2227--2235.


\bibitem{yan2015}
R.~Yan, Y.~Song, C.-T. Li, M.~Zhang, and X.~Hu, ``Opportunities or risks to
  reduce labor in crowdsourcing translation? characterizing cost versus quality
  via a pagerank-hits hybrid model.'' in \emph{IJCAI}, 2015, pp. 1025--1032.


\bibitem{yuen2015}
M.-C. Yuen, I.~King, and K.-S. Leung, ``Taskrec: A task recommendation
  framework in crowdsourcing systems,'' \emph{Neural Processing Letters},
  vol.~41, no.~2, pp. 223--238, 2015.


\bibitem{Zhang2025}
J.~Zhang, L.~Yi, X.~Gao, S.~S. Bhatti, T.~Yuan, and G.~Chen, ``Task scheduling
  mechanism for crowdsourcing in mobile social networks,'' \emph{IEEE
  Transactions on Mobile Computing}, pp. 1--15, 2025.


\end{thebibliography}
\end{document}